\documentclass[11pt]{article}
\usepackage{fullpage}
\usepackage{appendix}
\usepackage[pdftex]{graphicx}
\usepackage{subfigure}
\usepackage{va}
\usepackage[margin= 1in]{geometry}

\let\oldbibliography\thebibliography
\renewcommand{\thebibliography}[1]{%
  \oldbibliography{#1}%
  \setlength{\itemsep}{0pt}%
}

\usepackage{amsfonts, amsmath, amssymb, amsthm, constants,comment}
\usepackage[usenames]{color}
\usepackage{dsfont}
\usepackage{algorithm} 
\usepackage{algpseudocode}
\usepackage{mathtools}
\DeclarePairedDelimiter{\ceil}{\lfloor}{\rfloor}

\definecolor{darkred}{RGB}{100,0,0}
\definecolor{darkgreen}{RGB}{0,100,0}
\definecolor{darkblue}{RGB}{0,0,150}

\usepackage{hyperref}
\hypersetup{colorlinks=true, linkcolor=red, citecolor=blue, urlcolor=darkgreen}
\hypersetup{colorlinks=true, linkcolor=red, citecolor=blue, urlcolor=darkgreen}

\newtheorem{theorem}{Theorem}
\newtheorem{definition}{Definition}
\newtheorem{lemma}{Lemma}
\newtheorem{proposition}{Proposition}
\newtheorem{remark}{Remark}

\newtheorem{corollary}{Corollary}

\providecommand{\nor}[1]{\left\lVert {#1} \right\rVert}

\providecommand{\scal}[2]{\left\langle{#1},{#2}\right\rangle}

\providecommand{\scalT}[2]{\left\langle{#1},{#2}\right\rangle}

\newcommand{\R}{\mathbb R}
\providecommand{\scal}[2]{\left\langle{#1},{#2}\right\rangle}

\DeclareMathOperator{\E}{\mathbb{E}}
\def\argmax{\operatornamewithlimits{arg\,max}}
\def\argmin{\operatornamewithlimits{arg\,min}}

\title{\textbf{\large{ Quantization and Greed are Good:\\
One bit Phase Retrieval, Robustness and Greedy Refinements}}}

\author{\normalsize{\textsc{
Youssef Mroueh$^\star,\dagger$, Lorenzo Rosasco$^{\star,\ddagger}$ }}\\
\small \em $\star$  LCSL, Massachussetts Institute of Technology and Istituto Italiano di Tecnologia. \\
\small \em $\dagger$ CBCL, McGovern Institute, and Computer Science and Artificial Intelligence Lab, MIT,USA.\\
\small \em $\ddagger$ DIBRIS, Universit\`{a} di Genova, ITALY\\
{\small \tt  ymroueh,lrosasco@mit.edu}}

\date{\today}
\begin{document}

\maketitle

\begin{abstract}

In this paper, we study the problem of  robust phase recovery. We   investigate 
a novel approach based on extremely quantized (one-bit) measurements  and  a corresponding recovery scheme. 
The proposed  approach  has surprising robustness properties and, unlike currently available methods, allows to efficiently 
perform phase recovery from  measurements affected by  severe (possibly unknown) non linear perturbations, 
such as distortions (e.g. clipping). 
Beyond robustness,  we show how  our approach can be used within  greedy approaches based on alternating minimization. 
In particular,  we propose novel initialization schemes for the alternating minimization achieving favorable convergence properties 
with improved sample complexity. 

\end{abstract}
\section{Introduction}
The phase recovery problem can be modeled as the problem of  reconstructing a $n$-dimensional  complex  vector $x_0$ given 
only the magnitude of $m$ phase-less linear measurements. Such  a problem arises for example in X-ray crystallography \cite{Harrison93,Liu08}, diffraction imaging \cite{Bunk07,Rodenburg} or microscopy \cite{Miao08},  where one can only measure the intensities of the incoming waves, and  wishes to recover the lost phase in order to be able to reconstruct the desired object. 

In practice, phase recovery is often tackled via  greedy algorithms \cite{Gerchberg72,Fienup82,Griffin84} which
 typically lack convergence guarantees.  Recently,   approaches based on convex relaxations, namely PhaseLift in \cite{Candes,Demanet}, and Phase cut in \cite{mallat}, have been proposed and analyzed.  These latter methods can be solved by  Semi Definite Programing (SDP),
 and   allow the exact and stable recovery   of the signal (up to a global phase) from $O(n)$ measurements. 
 A different approach have been recently considered in \cite{AM},  where it is shown that a greedy alternating minimization, akin to those in \cite{Gerchberg72,Fienup82,Griffin84}, can be shown to geometrically converge to the true vector
 $x_0$ if $O(n\log^3(n))$ measurements  are given. Indeed, alternating minimization algorithms are known to be extremely sensitive  
 to the initialization and  a suitable initialization is the key of the analysis in \cite{AM}.
While alternating minimization approaches provide a solution only up-to a given accuracy, they often  have very good practical 
performances when compared to convex methods \cite{AM}, with dramatic computational advantages \cite{AM}. The solution of the SDP 
in convex approaches is  computationally expensive and  needs to be close to a rank one matrix for  tight recovery (which is  rarely encountered in practice \cite{mallat}). Indeed,  some greedy refinement of the SDP solution is often considered  \cite{mallat}.

In this paper, we propose and investigate a phase recovery approach based on extremely quantized measurements 
and a corresponding recovery procedure. In particular, we study the properties of the proposed method towards the following questions:
\begin{enumerate}
\item \textit{Robustness: Is it possible to efficiently perform  phase recovery, when  the   measurements
are corrupted by sever perturbations such as non linear distortions or   stochastic  noise?}
\item \textit{Refinements of Alternating Minimization: 
Are there initialization strategies for the alternating minimization approach that allow better sample complexity?}
\end{enumerate}
Robustness to noise and distortions, such as clipping of the intensities or imperfections in Fourier optics such as multiple scattering \cite{distortion} is a desirable property for a phase  retrieval algorithm.
At first this task might seem hopeless since current approaches to phase recovery are based on  measurements 
magnitudes which might be completely altered by distortions or if the signal-to-noise ratio is very poor.  
In fact, we prove the somewhat surprising fact that  phase retrieval is still possible, as long as the perturbations 
 preserve (on average) the  {\em ranking} of the measurements intensities.
 Indeed, key to our approach is considering a suitable quantization scheme based on comparing pairs of phase-less measurements:  
only the ranking of each measurement  pairs becomes important, rather than the intensity values themselves. 
Using these extremely quantized (one-bit) measurements, we show that recovery is possible  as soon as  $O(n\log n)$ pairs of measurements are available.  The corresponding  recovery procedure reduces to a  maximum eigenvalue problem ({\em 1bitPhase}) which can be efficiently solved, for example using the power method. Our approach is inspired by the growing field of one-bit compressive sensing \cite{1bitCS,Vershynin,LPCS, CSLaska}.

Beyond robustness, we show how the nature of the one-bit phase less measurements can be used to obtain better results 
for alternating minimization. We show that the solution of one-bit phase retrieval can be used to initialize alternating minimization 
to obtain the same convergence results in \cite{AM} from only $O(n\log n)$ measurements. Finally, we study a further initialization (weighted one-bit phase retrieval), which is a hybrid between the one in \cite{AM} and the one  provided by the one-bit approach.


The rest of the paper is organized as follows.  In Section \ref{sec:back}, we discuss some background and previous results.
In Section \ref{sec:Res}, we sketch our main results and techniques.  In Section \ref{sec:Sensing and Recovery}, we introduce and analyze the One-Bit Phase Retrieval approach. In Section \ref{sec:WeOneBit} we introduce weighted one bit phase retrieval that uses both quantized and un-quantized measurements. In Section \ref{sec:Greed} we show how one-bit phase retrieval algorithms can be used to initialize the alternating minimization approach  to  get a better  sample complexity.
We provide a theoretical analysis of our approach in Section \ref{theory}. Finally,  in Section \ref{sec:comp} we discuss  some computational aspects and present some numerical results.\\ 
 
 \noindent \textbf{Notations:}
 For $z\in \mathbb{C}$, $|z|^2$ is  squared  complex modulus of $z$.
For $a,a' \in \mathbb{C}^n$, $\scalT{a}{a'}$ is the complex dot product in $\mathbb{C}^n$. For $a \in \mathbb{C}^n, a^*$ is the complex conjugate and $||a||_{2}$ or simply $||a||$ is the norm $2$ of $a$.
Let $A$ a complex hermitian matrix in $\mathbb{C}^n$, $||A||_{F}$ denotes the Frobenius norm of $A$,  $||A||$ denotes the operator norm of $A$, $Tr(A)$ denotes the trace of A. Throughout the paper, we denote by $c,C$ positive absolute constants whose values may
change from instance to instance. 
\section{Background and Previous Work}\label{sec:back}

In this section,  we formalize the problem of  recovering a signal from  phase-less measurements  and discuss previous results.
Throughout this section,  and the  rest of the paper,  we consider  measurements defined  by independent and identically distributed Complex Gaussian sensing vectors, 
\begin{equation}\label{GaussVectors}
a_i \in \mathbb{C}^n, \quad \quad a_i \sim \mathcal{N}(0,\frac{1}{2}I_n)+i \mathcal{N}(0,\frac{1}{2}I_n),  \quad i=1\dots m.
\end{equation} 
The (noiseless) phase recovery problem is defined as follows.
\begin{definition}[Phase-less Sensing and Phase Recovery]
Suppose  phase-less sensing measurements 
\begin{equation}\label{Meas}
b_i=|\scalT{a_i}{x_0}|^2 \in \mathbb{R}_{+}, \quad \quad i=1\dots m,
\end{equation}
are given for $x_0 \in \mathbb{C}^n$,  where $a_i,~ i=1, \dots, m$ are random vectors as in \eqref{GaussVectors}.
The phase recovery problem is
\begin{equation}
\begin{aligned}
& \underset{}{\text{find}~x},\quad
& \text{subject to}\quad 
|\scalT{a_i}{x}|^2=b_i,\quad i=1\dots m. 
\end{aligned}
\label{eq:np}
\end{equation}
\end{definition}

\noindent The above  problem is non convex and in the following we recall recent approaches to provably and efficiently recover $x_0$
from a finite number of measurements.
%

\noindent \textbf{SDP (Convex) Relaxation and  PhaseLift.} 
The PhaseLift approach   \cite{Candes} stems from the  observation that $|\scalT{a_i}{x}|^2=Tr(a_ia^*_ixx^*),$ so that 
 if we let $X=xx^*$, Problem \ref{eq:np} can be written as, 
\begin{equation}
\begin{aligned}\label{PLift}
& \underset{}{\text{find}~X,}
& \text{subject to}\quad
&Tr(a_ia_i^*X)=b_i, \quad i=1\dots m, 
& \quad X \succeq 0, \quad rank(X)=1.
\end{aligned}
\end{equation}
While the above formulation is still non convex (and in fact combinatorially hard because of the rank constraint), a convex relaxation can be 
obtained noting that Problem \ref{PLift} can be written as   a rank minimization problem over the positive semidefinite  cone,
\begin{equation}
\underset{X}{\text{min}}\quad rank(X),\quad  \text{subject to}\quad Tr(a_ia_i^*X)=b_i, ~~i=1\dots m,  \quad  X \succeq 0,
\end{equation}
and then  considering the   trace as a surrogate for the rank  \cite{Candes},
\begin{equation}
\underset{X}{\text{min}}\quad Tr(X),\quad  \text{subject to}\quad Tr(a_ia_i^*X)=b_i, ~~i=1\dots m,  \quad  X \succeq 0.
\end{equation}
Indeed, the above problem is  convex and can be solved via semidefinite programming (SDP). 
Intestingly, a different  relaxation is obtained  in \cite{Demanet}
by  ignoring the rank constraint in Problem \ref{PLift}. 
The results in \cite{Candes,Demanet}  show that, with high probability,  the solution $\hat{X}_m$ obtained via either one of the above  relaxations can recover $x_0$ {\em exactly}, i.e. $\hat{X}_m=x_0x_0^*$, as soon as  $m\geq c n\log n$.  In fact,  the latter requirement can be further improved to $m\geq c n$ \cite{candesn}.  
If the  measurements are  corrupted by  noise, namely
\begin{equation}\label{NoisyMeas}
b_i=|\scalT{a_i}{x_0}|^2+w_i,\quad i=1\dots m,
\end{equation} 
the PhaseLift approach can be adapted \cite{candesn} by considering 
\begin{equation}\label{NoisyPLift}
\underset{X}{\text{min}} \sum_{i=1}^m \left| Tr(a_ia_i^*X)-b_i\right|,\quad \text{subject to} \quad \quad \quad X \succeq 0.
\end{equation}
The above problem is  convex and can again be solved via an SDP approach.   The properties of its solution have been studied in 
\cite{candesn}  for deterministic noise  $$||w||_{1}\leq \delta, \quad w=(w_1, \dots, w_n) \in \mathbb R _+ ^m, $$ 
where it is shown that the solution $\hat{X}_m$ of \eqref{NoisyPLift} 
 satisfies $||\hat{X}_m-x_0x_0^*||_{F}\leq c \delta/m$, as soon as $m\geq c n$. 
Moreover, the leading eigenvector $\hat{x}_m$ of $\hat{X}_m$ satisfies 
$||\hat{x}_m-e^{i\phi}x_0||_{2}\leq c \min (||x_0||_2,\frac{\delta}{m||x_0||_2})$,
where $\phi $ is a global phase in $[0,2\pi]$. Most importantly, the latter results suggests that 
 $x_0$ can be recovered  considering  the leading eigenvector of $\hat{X}_m$.

As  mentioned in the introduction, while powerful, the convex relaxation approach incur in cumbersome computations-- see Table \ref{tab:comparison}, 
and in practice non convex approaches based on greedy alternating minimization (AM)
\cite{Gerchberg72,Fienup82,Griffin84} 
are often used. The convergence properties of the latter methods  depend heavily  on the initialization and only 
recently \cite{AM}  they have been shown to globally converge (with high probability) if provided with a suitable initialization.
We next briefly review these latter results, which we further discuss and extend in Section \ref{sec:Greed}. 
%
%
 
\noindent\textbf{Phase Retrieval via Suitably Initialized  Alternating Minimization.}  
Let $A$ be the matrix defined by $m$ sensing vectors as in \eqref{GaussVectors}  and $B=Diag(\sqrt{b})$, where $b$ is the vector of measurements as in \eqref{Meas}. Then,
$$
Ax_0=Bu_0,
$$
for $u_0=Ph(Ax_0)$ with  $Ph(z)=\left(\frac{z_1}{|z_1|}, \dots \frac{z_m}{|z_m|}\right)$,  
$z \in \mathbb{C}^n$. 
The above equality suggests the following   natural approach  to recover $(x_0,u_0)$,
\begin{equation}
\underset{x,u}{\text{min}}
||Ax-Bu||_2^2, \quad \text{subject to}
\quad 
 |u_i|=1, \quad i=1\dots m, \\
\end{equation}
The above problem is non-convex  because of the constraint on $u$ and  the AM approach 
(Algorithm \ref{AltMinPhase}) consists in optimizing  $u$, for a given $x$, and then optimizing $x$ for a given $u$.
It is easy to see that for a given   $x$,  the optimal $u$ is simply
$u=Ph\left(Ax\right),$ and for a given $u$,the optimal  $x$ is the solution of  a least squares problem. 

The key result in \cite{AM} shows that if such an iteration is initialized with  maximum eigenvector of the matrix  
\begin{equation}\label{Covb}
\hat{C}_m=\frac{1}{m}\sum_{i=1}^mb_i a_ia_i^*
\end{equation}
and $m\ge C n(\log n)^3$, 
then the solution $\hat{x}_m$ of the alternating minimization globally converge (with high probability) to the true vector $x_0$.
Moreover for a given accuracy $\epsilon\in [0,1]$,   if  
\begin{equation}\label{AM_SampComp}
m\ge c(n(\log^3(n)+\log(\frac{1}{\epsilon})\log(\log(\frac{1}{\epsilon})) )),
\end{equation} 
then  $||\hat{x}_m-e^{i\phi}x_0||_{2}\leq \epsilon$. 

The following  two remarks will be useful in the following. 
First,   a key observation, motivating  the above initialization (called SubExpPhase in the following),  is the fact that the expectation of $\hat{C}_m$ can be shown to satisfy  $\mathbb{E}(\hat{C}_m)=x_0x_0^*+I$. Indeed,  the proof in \cite{AM}  (see Section \ref{sec:Greed}) relies   
on the concentration properties of the random matrix $\hat{C}_m$ around its expectation \cite{VershyninReview}. 
Second, it is useful to note that these latter  results crucially depend on a bound on the norm of $b_i a_ia_i^*$ for $i=1, \dots,m$. Indeed, it is this latter bound the main cause of the  poly-logarithmic term in the sample complexity \eqref{AM_SampComp}, since the $b_i$'s are 
sub-exponential random variables.

\begin{algorithm}[H]
 \begin{algorithmic}[1]
 \Procedure{AltMinPhase}{$A,b$}
 \State \textbf {Initialize} $x$ .
 \For{k=1\dots} 
  \State $u \gets Ph(Ax)$
 \State  $ x \gets \arg\min ||Ax-Bu ||^2_{2}$
 \EndFor
 \State \textbf{return} $x$ 
 \EndProcedure
 \end{algorithmic}
 \caption{AltMinPhase}
 \label{AltMinPhase}
\end{algorithm}

\section{Summary of Our Main Results and Techniques}\label{sec:Res}
In this paper,  we propose and study a quantization scheme and a corresponding recovery procedure. 
In particular we investigate the properties of our approach towards: 1) the phase recovery problem from severely 
perturbed measurements,  and   2) the improvement of the AM approach discussed in the previous section.

\subsection{Phase Recovery from Severely Perturbed Measurements}
We investigate the phase recovery problem in the case in which we have at disposal measurements of the form 
\begin{equation}\label{DistMeas}
b_i=\theta(|\scalT{a_i}{x_0}|^2), \quad i=1, \dots, 2m,
\end{equation}
where $a_i$ are sensing vectors as in \eqref{GaussVectors} and  $\theta$ is a {\em possibly unknown}  rank preserving transformation. 
In particular we are interested to situations where $\theta$  models a distortion, e.g. $\theta(s)=\tanh(\alpha s)$, $\alpha\in \mathbb R_+$, 
or an additive noise  $\theta(s)=s+\nu$,   where $\nu$ is a  stochastic  noise, such as an exponential or Poisson noise.
As we noted before the  recovery problem from severly perturbed intensity values seems hopeless, and indeed  the key in  our approach is a quantization scheme based on comparing  pairs of phase-less measurements. 
More precisely for each pair $b^1, b^2$ of measurements of the form  \eqref{DistMeas} we define  $y\in \{-1, 1\}$
as $y=sign(b^1-  b^2)$.  This  {\em one-bit}  quantization scheme  draws inspiration from ideas 
in one-bit compressive sensing \cite{1bitCS,Vershynin,LPCS, CSLaska}, but the fundamental difference is that our 
approach crucially depends on the  comparison of two measurements: one-bit measurements involve the spacing between the order statistics  of exponentially distributed random variables $b^1$ and $b^2$. Indeed, this will be a key fact in our analysis.
While phase-recovery from one-bit phase-less measurements is in general a hard problem (see Section \ref{sec:Sensing and Recovery}),
we propose to consider a relaxation which reduces to a maximum eigenvalue problem induced by the matrix
\begin{equation}\label{Covy}
\hat{C}_m=\frac{1}{m}\sum_{i=1}^{m}  y_i(a^{1}_ia^{1,*}_i-a^{2}_ia^{2,*}_i).
\end{equation}
When compared to  \eqref{Covb}
we see that the $m$ phase-less measurements $b_i$ are replaced 
by their quantized counterpart $y_i$ (obtained from $2m$ phase-less measurements), 
and the term given by sensing vectors is now given  by {\em pairs} of sensing vectors.
Indeed, we prove in Section \ref{theory} that the expectation of $\hat{C}_m$ satisfies 
$\mathbb E  \hat{C}_m=\lambda x_0x_0^*$, where $\lambda$ is a suitable constant
which depends on $\theta$ and plays the role of a signal-to-noise ratio. Indeed, by studying the concentration 
properties of the matrix $\hat{C}_m$, we show  that, for a given accuracy $\epsilon\in [0,1]$,  if 
$O( \frac{ n\log n }{\epsilon^2 \lambda})$ pairs of measurements  are available, then 
 the solution of the above  maximum eigenvalue problem satisfies 
$$||\hat{x}_m-x_0e^{i\phi}||^{2}\leq \epsilon,$$  
where  $\phi \in [0,2\pi]$ is a global phase. 
It is worth noting here that the signal can be recovered up to a scaling factor from one bit measurements, but this is not a problem since our goal is to recover the missing phase.
\subsection{One-Bit Phase Retrieval and  Alternating Minimization}

A key difference between the matrix in Eq. \eqref{Covb} and the one in Eq.  \eqref{Covy} is that one-bit measurements are bounded
and lead to improved concentration results. This motivates considering the effect of using the solution of the one-bit phase retrieval 
to initialize the alternating minimization procedure considered in \cite{AM}.
Indeed, leveraging results from \cite{AM}, 
we prove in Section \ref {sec:Greed} that,  provided 
with the one-bit retrieval initialization,  
the alternating minimization algorithm globally converges (with high probability) to the true vector $x_0$, and  if  
\begin{equation}\label{AM_SampComp1Bit}
m\ge c(n(\log n+\log(\frac{1}{\epsilon})\log(\log(\frac{1}{\epsilon})) )),
\end{equation} 
then  $||\hat{x}_m-e^{i\phi}x_0||_{2}\leq \epsilon$. 
Comparing to \eqref{AM_SampComp},  we see that the sample complexity  depends  now only on a  logarithmic term.
Quantization can be seen as playing  the role of a preconditioning that enhances the sample complexity of the alternating minimization.
Further, we note that  it is possible  to achieve similar results, see Table \ref{tab:comparison},  considering   a different initialization obtained via a weighted one-bit approach which combines 
quantized and un-quantized measurements, see Section \ref{sec:WeOneBit}. 

\begin{table}[H]
  \begin{center}
    \begin{tabular}{ | c | c | c |}
      \hline
								& Sample complexity & Comp. complexity  \\ \hline
          PhaseLift	& $O(n)$
					& $O(n^3/\epsilon^2)$ \\ \hline
      PhaseCut	& $O(n)$
					& $O(n^3/\sqrt{\epsilon})$\\ \hline
   SubExpPhase+AM 	& $ O(n \left(\log^3 n + \log \frac{1}{\epsilon} \log \log \frac{1}{\epsilon}\right))$
					&  $O(n^2 \left(\log^3 n + \log^2 \frac{1}{\epsilon} \log \log \frac{1}{\epsilon}\right))$ \\ \hline

     1bitPhase+AM & $O(2n\left(\log(n)+ \log \frac{1}{\epsilon} \log \log \frac{1}{\epsilon}\right)))$     &		$O(n^2 \left(\log n + \log^2 \frac{1}{\epsilon} \log \log \frac{1}{\epsilon}\right))$	\\ \hline		Weigthed1bitPhase+AM & $O(2n\left(\log(n)+ \log \frac{1}{\epsilon} \log \log \frac{1}{\epsilon}\right)))$     &		$O(n^2 \left(\log n + \log^2 \frac{1}{\epsilon} \log \log \frac{1}{\epsilon}\right))$	\\ \hline

      \end{tabular}
      \caption{Comparison of the sample and computational complexity of different phase retrieval schemes.}
      \label{tab:comparison}
  \end{center}
\end{table}

\section{One-Bit Phase Retrieval}\label{sec:Sensing and Recovery}
In this section, we set up the one-bit approach to phase-retrieval and state our main results on robust phase recovery.
\subsection{ Quantization and Recovery}
Unlike  in compressive sensing,  in our context the intensities are non negative, and thus we cannot rely on only  one measurement to build a quantizer. 
\begin{definition}[One-bit quantizer]  \label{def:quantizer} 
Let $A=(a^{1},a^{2})$, where $a^{1},a^{2}$ are \text{i.i.d.} 
complex  Gaussian vectors 
$\mathcal{N}(0,\frac{1}{2}I_{n})+ i\mathcal{N}(0,\frac{1}{2}I_{n})$.
For $x_0 \in \mathbb{C}^n$, a one bit quantizer is given by 
$$Q_{A}:\mathbb{C}^n\to \{-1,1\},\quad Q_{A}(x_0)=sign\left( |\scalT{a^1}{x_0}|^2-|\scalT{a^2}{x_0}|^2\right).$$
\end{definition}
\begin{definition}[Quantized phase-less measurements] Let $\{A_i=(a^1_i,a^2_i)\}_{1\leq i\leq m}$, be $2m$ i.i.d. gaussian complex vectors in $\mathbb{C}^n$, and $Q_{A_i}(x_0)$ as in Def \ref{def:quantizer} . The Quantized Phase-less sensing is given by $\mathcal{Q}:\mathbb{C}^{n}\to \{-1,1\}^m$ , $\mathcal{Q}(x_0)=(Q_{A_1}(x_0),\dots,Q_{A_m}(x_0))$.
\label{def:sensing}
\end{definition}
 In this paper, we are interested in recovering $x_0$ from its  quantized phase-less measurements $y= (y_1\dots y_m)=\mathcal{Q}(x_0)=(Q_{A_1}(x_0),\dots,Q_{A_m}(x_0))$.
It is easy to see that the recovery problem has the form,
\begin{equation}
 \underset{}{\text{find}~~ x},\quad  \quad \text{subject to} \quad y_i\left(\left|\scalT{a^{1}_i}{x}\right|^2 - \left|\scalT{a^{2}_i}{x}\right|^2\right)\geq0,\quad i=1\dots m, \quad ||x||^2_2=1.
\label{eq:biPhase}
\end{equation}
Indeed, as in one-bit compressive sensing, we cannot hope to recover the norm of the vector from inequality constraints, hence 
the  norm one constraint. Problem \eqref{eq:biPhase} can be equivalently written as the following quadratically constrained problem, 
\begin{equation}
 \underset{}{\text{find}~~ x},\quad  \quad \text{subject to} \quad x^*y_i(a^{1}_ia^{1,*}_i-a^{2}_ia^{2,*}_i)x\geq0,\quad i=1\dots m, \quad ||x||^2_2=1.
\label{eq:QPhase}
\end{equation}
The above problem is a non-convex Quadratically Constrained Quadratic Program  (QCQP) and can be shown to be NP-hard in general \cite{qcqp}.
 We propose to consider the following relaxation, 
   \begin{equation}
 \underset{x}{\text{max}}\quad  x^* \left(\frac{1}{m}\sum_{i=1}^{m}  y_i(a^{1}_ia^{1,*}_i-a^{2}_ia^{2,*}_i)\right)x,\quad  \quad \text{subject to}  \quad ||x||^2_2=1.
\label{eq:MaxPhase}
\end{equation}
The 1BitPhase problem is obtained  noting that the above problem can be rewritten as the 
the maximum eigenvalue problem,
 \begin{equation}
\begin{aligned}
& \underset{}{\text{$\max_{x~\text{s.t.}~||x||_2=1}$}~~ x^*\hat{C}_mx},
\end{aligned}
\label{eq:MaxPhase}
\end{equation}
defined by the matrix 
$$ \hat{C}_m=\frac{1}{m}\sum_{i=1}^{m}  y_i(a^{1}_ia^{1,*}_i-a^{2}_ia^{2,*}_i).$$

\noindent As we comment in the following remark the 1BitPhase approach is  inspired by one bit-compressed sensing. 

\begin{remark}[Quantization in Compressive Sensing: One bit CS]
Non-linear, quantized measurements have been recently considered in the context of one-bit compressive sensing\footnote{See  {\small \texttt{http://dsp.rice.edu/1bitCS/}} for an exhaustive list of references.}. 
Here, binary (one-bit) measurements are obtained  applying, for example, the ``sign'' function to linear measurements. More precisely, given 
$x_0\in \R^{n}$,  a measurement vector is given by  $y=(y_1, \dots, y_m)$, where $y_i= \text{sign}(\scal{a_i}{x})$ with  $a_i\sim {\mathcal N} (0, I_{n})$ independent Gaussian random  vectors, for $i=1,\dots, m$. It is possible to prove, see e.g.  \cite{Vershynin},   that, for a signal $x_0 \in K\cap  \mathbb B ^{n}$ ($\mathbb{B}^{n}$ is the unit ball in $\R^n$),
 the solution $\hat x_m$ to the problem
\begin{equation}\label{ERM}
\max_{x\in K\cap  \mathbb B ^{n}}\sum_{i=1}^m y_i\scal{a_i}{x},
\end{equation}
satisfies $\nor{\hat x_m - x_0}^2\le  \frac{\delta}{\sqrt{\frac{2}{\pi}}}$,  $\delta>0$, with high probability, 
as long as $m\ge  C\delta^{-2} \omega(K)^2$ \cite{Vershynin}. Here  $\omega(K)=\E \sup_{x\in K-K} \scal{w}{x}$  denotes the Gaussian mean width of $K$.
The 1BitPhase approach shows that  a relaxation  similar to problem \ref{ERM} allows 
to perform  phase recovery for a suitably defined quantization of phase-less linear measurements.
\end{remark}

Before studying the recovery guarantees for the solution of 1BitPhase we briefly discuss a geometric intuition 
underlying the method.

\subsection{Geometric Intuition}

\begin{figure}[t]
\begin{center}
\noindent\includegraphics[scale=0.3]{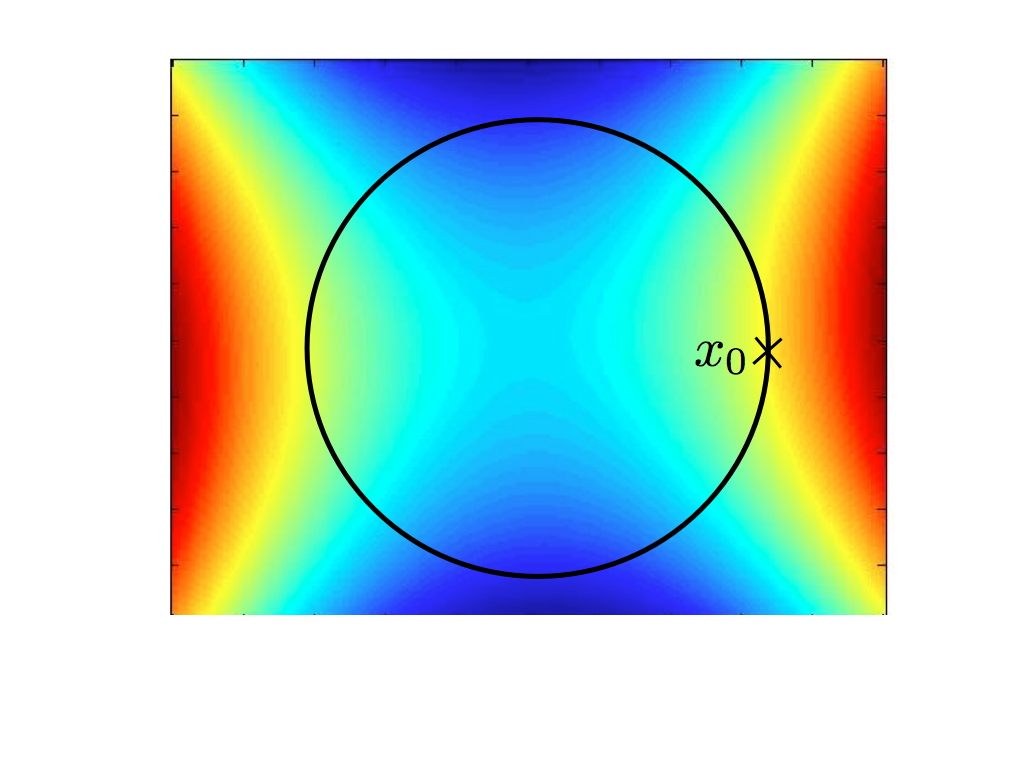}
\end{center}
\caption{The level sets of  the objective of problem the 1bitPhase problem (in red maximal values, in blue minimal values).}
\label{fig:obpr}
\end{figure}

\noindent To understand the geometric intuition of one bit phase retrieval, we first consider  the feasibility problem  \eqref{eq:biPhase}.
Each measurements pair $(a^1_i,a^2_i)$ defines a hyperbolic paraboloid $z=x^* (a^1_ia^{1,*}_i- a^2_i a^{2,*}_i)x$. The feasible zone is defined by the constraints $y_i(x^* (a^1_ia^{1,*}_i- a^2_i a^{2,*}_i)x)\geq 0$, which enforce the geometric consistency with the bit $y_i$. In other words each constraint says that $x_0$ is in the region of the space where the sign  of the corresponding 
  hyperbolic paraboloids is  $y_i$.  Note that  each  hyperbolic paraboloid is symmetric with respect to the origin, thus  the feasible region 
  is also symmetric with respect to the origin. When we add more constraints we have  that $x_0$ lies in the intersection of such symmetric feasible regions and  the unit sphere. This intersection is also symmetric, thus we can solve the phase retrieval up to global sign flip (in the real valued case). As we mentioned before,  the feasibility problem is a non convex QCQP which  is  NP hard in general and our relaxation \eqref{eq:MaxPhase} can be seen as requiring the  geometric consistency with one bit measurements on average, rather than individually 
  as in the feasibility problem.
In  figure \ref{fig:obpr}  we plot the level sets of the objective function of the  $1$bitPhase Problem  \eqref{eq:MaxPhase}  
for $x_0=(1,0)$. We see that the objective function achieves its maximum values (in red) in a symmetric region. This region intersects with the sphere in two regions close to  the points $x_0$ and $-x_0$. Thus we are able to recover the phase up to global sign and scaling from single bit measurements.
  
\subsection{ One bit Phase Retrieval from Distorted and Noisy measurements}
In the following we assume that the intensities are undergoing an unknown non linearity  $\theta$, that is  we observe,
$$(b^1_i,b^2_i)=(\theta(|\scalT{a^1_i}{x_0}|^2), \theta(|\scalT{a^2_i}{x_0}|^2), \quad i=1 , \dots , m.$$
 Thus the quantized measurements are, 
\begin{equation}\label{eq:model}
y_i=Q^{\theta}_{A_i}(x_0)=sign\left(\theta(|\scalT{a^1_i}{x_0}|^2)-\theta(|\scalT{a^2_i}{x_0}|^2\right), \quad i=1\dots m.
\end{equation}
For instance clipping can  be modeled by a sigmoid,
$$\theta(z)=\tanh(\alpha z), z>0.$$
where  the parameter $\alpha$ controls how severe is the distortion.  An additive noise (before quantization) can be modeled by,
\begin{equation}
\theta(z)=z+\nu, 
\label{eq:noispr}
\end{equation}
where $\nu \sim Exp(\gamma)$ is a stochastic exponential noise with mean $\mu=\frac{1}{\gamma}$ and variance $\sigma=\frac{1}{\gamma^2}$.\\
When the intensities are contaminated with poisson noise we have,
\begin{equation}
\theta(|\scalT{a}{x_0}|^2)=\mathcal{P}_{\eta}\left(|\scalT{a}{x_0}|^2\right) ,
\label{eq:1bitCDPoisson}
\end{equation}
where $\mathcal{P}_{\eta}$ is a poisson noise, such that :\\
$$\text{For } z,\eta >0,\quad  \theta(z)=\mathcal{P}_{\eta}(z)=p ,\quad  \text{ where } p\sim Poisson\left( \frac{z}{\eta}\right).$$
We shall make one assumption on the non linearity $\theta$,
\begin{equation}{\label{eq:lambda}}
\lambda=\mathbb{E}(sign(\theta(E_1)-\theta(E_2))(E_1-E_2))>0,
\end{equation}
where $E_1$, $E_2$ are two independently distributed exponential random variables. To see why this assumption is natural, 
notice that $|\scalT{a}{x_0}|^2\sim Exp(1)$  if $a \sim \mathcal{C}\mathcal{N}(0,I_n)$ and $||x_0||=1$, thus
$$\mathbb{E}(y_i(|\scalT{a^1_i}{x_0}|^2-|\scalT{a^2_i}{x_0}|^2 ))=\mathbb{E}(sign(\theta(E_1)-\theta(E_2))(E_1-E_2))=\lambda>0.$$
Then the above assumption simply means that the one bit measurements preserve robustly the ranking of the intensities. 
For example, such an  assumption is trivially satisfied whenever $\theta$ is increasing. 
In case $\theta(z)=z$, $\lambda$ achieves its maximal value,
$$\lambda=\mathbb{E}\left(sign(E_1-E_2)(E_1-E_2)\right)=\mathbb{E}\left|E_1-E_2\right|=1,$$
since $\left|E_1-E_2\right|\sim Exp(1)$. 
For different models of observation, the value of $\lambda$ is given in Lemma~\ref{lem:SNR}, which shows that 
 that $\lambda$ plays the role of a signal to noise ratio.
\begin{figure}[H]\label{fig:lambper}
\noindent\includegraphics[scale=0.4]{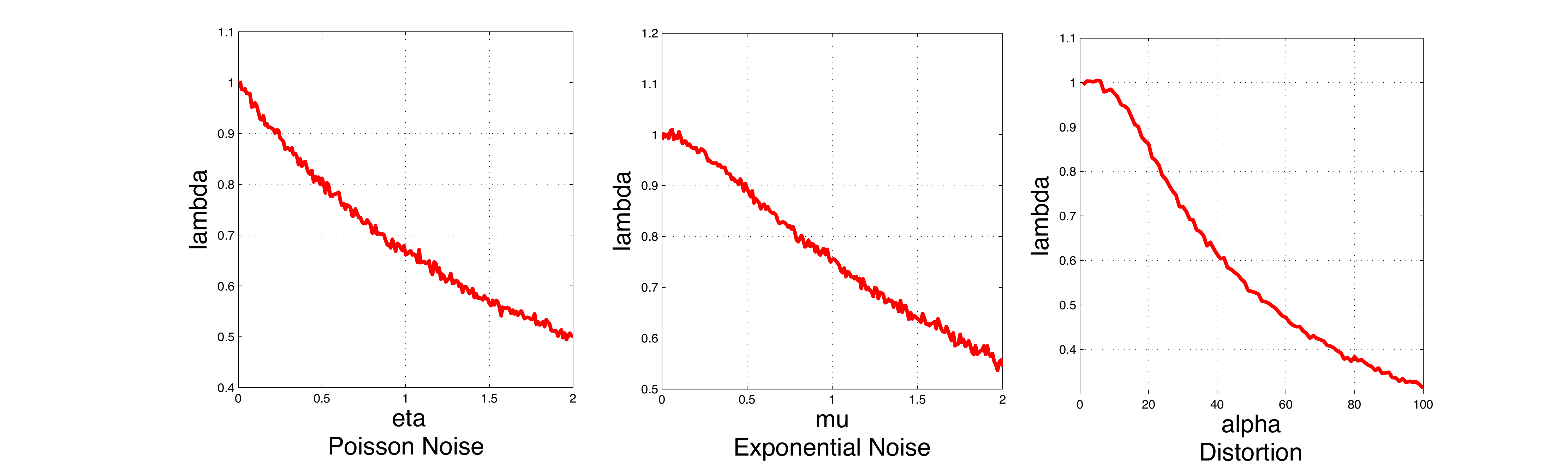}
\caption { $\lambda$ as a signal to noise ratio: $\lambda$ versus different parameters of various observation model $\theta$. $\lambda$ decreases as noise and distortion levels increase.}
\end{figure}
\noindent We see in Figure \eqref{fig:lambper} that $\lambda$ indeed decreases as the level of noise and distortion increases.
\subsection{Main Results for  One Bit Phase Retrieval}
The following theorem describes the recovery guarantees for the solution   $\hat{x}_m$  of problem $1$bitPhase \eqref{eq:MaxPhase}.
\begin{theorem}[One bit Recovery]
For $x_0 \in \mathbb{C}^n, ||x_0||=1$. Assume $y_1\dots y_m$, follows the model  given in \eqref{eq:model}. Then for any $\epsilon \in [0,1]$, we have with a probability at least $1- O(n^{-2}) $, 
$$\text{for } m \geq \frac{ C }{\epsilon^2 \lambda} n\log(n), \quad ||\hat{x}_m-x_0e^{i\phi}||^{2}\leq \epsilon,$$
where  $\phi \in [0,2\pi]$ is a global phase and $\lambda$ is given in \eqref{eq:lambda}.
\label{theo:main}
\end{theorem}
For the simple model model where $\theta(z)=z$, $\lambda=1$. Theorem \ref{theo:main} implies that if $m=O(n\log(n))$ (so that the total number of measurements is $2m$), then $\hat{x}_m$ is an $\epsilon$-estimate of $x_0$, up to a global phase $\phi$.
In Corollary \ref{theo:Noise} we specify the above theorem to the noisy model \eqref{eq:noispr}.
\begin{corollary}[One bit Recovery/ Noise]
For $x_0 \in \mathbb{C}^n, ||x_0||=1$, and $\epsilon>0$. Assume $y_1\dots y_m$, follows the noisy model  given in \eqref{eq:model}, for $\theta(z)=z+\nu, \nu \sim Exp(\gamma)$.
Where $\nu$ is an exponential noise with variance $\sigma=\frac{1}{\gamma^2}$. 
 Then for any $\epsilon \in [0,1]$, we have with a probability at least $1- O(n^{-2}) $, 
$$\text{for } m \geq \frac{C  }{\epsilon^2}\frac{(1+\sqrt{\sigma})^2}{1+2\sqrt{\sigma}} n\log(n), \quad ||\hat{x}_m-x_0e^{i\phi}||^{2}\leq \epsilon,$$
where  $\phi \in [0,2\pi]$ is a global phase. 
\label{theo:Noise}
\end{corollary}
\noindent In other words, under an exponential noise we have: $$||\hat{x}_m-x_0e^{i\phi}||^{2}\leq C\sqrt{\frac{ n\log(n)}{m}\frac{(1+\sqrt{\sigma})^2}{1+2\sqrt{\sigma}} }.$$
A similar result holds for Poisson noise for a different value of $\lambda$ given in Lemma \ref{lem:SNR}.\\
Beyond robustness to noise, another desirable feature for phase retrieval from phase-less  measurements, is the robustness to distortions of the values of intensities.
Is it possible to retrieve the phase from intensities values that are undergoing clipping for instance? 

\begin{corollary}[One bit Recovery/ Distortion]\label{cor:Distortion}
For $x_0 \in \mathbb{C}^n, ||x_0||=1$, and $\epsilon>0$. Assume $y_1\dots y_m$, follows the noisy model  given in \eqref{eq:model}, for $\theta(z)=\tanh(\alpha z), \alpha >0$.
 Then for any $\epsilon\in [0,1]$, we have with a probability at least $1-O(n^{-2})$, 
$$\text{for } m \geq \frac{C }{\epsilon^2} \frac{ n\log(n)}{\lambda(\alpha)}, \quad ||\hat{x}_m-x_0e^{i\phi}||^{2}\leq \epsilon,$$
where  $\phi \in [0,2\pi]$ is a global phase. 
$\lambda(\alpha)=\mathbb{E}(|E_1-E_2|sign(1-\tanh(\alpha E_1)\tanh(\alpha E_2)))$ is a decreasing function in $\alpha$.
\end{corollary}

The proof of the above results  follow from a simple combination of  Propositions \ref{lem:compeq},\ref{lem:empiprocess}, and \ref{lem:radav} given in  Section \ref{theory}.

\subsection{Weighted One Bit Phase Retrieval}\label{sec:WeOneBit}
The boundedness of one bit phase measurement $y$ is appealing as it ensures better sample complexity.
In this section, we  ask the question of whether similar results can be obtained combining 
the un-quantized measurements $(b^1_i= |\scalT{a^1_i}{x_0}|^2,b^2_i=|\scalT{a^2_i}{x_0}|^2),i=1\dots m $, and the quantized measurements $y_i=sign(b^1_i-b^2_i)$.

\noindent We have therefore to keep in mind that we need to formulate the problem in such way the random variables depending on $(b^1_i,b^2_i)$ are bounded. 
Then we consider  $$R^1_i =\frac{b^1_i}{b^1_i+b^2_i} \text{ and }R^2_i=\frac{b^2_i}{b^1_i+b^2_i}\quad i=1\dots m.$$
$(R^1_i,R^2_i)$ take values  in $[0,1]^2$(hence bounded), moreover they are Beta distributed $Beta(1,1)$ , that is  the uniform distribution $unif [0,1]$ (Lemma \ref{lem:unif}).
Then, we can consider  the following problem, 
\begin{equation}
\begin{aligned}
& \underset{}{\text{find}~~ x}\\
& \text{subject to}\\
&y_i\left(R^1_{i}\left|\scalT{a^{1}_i}{x}\right|^2 - R^2_{i}\left|\scalT{a^{2}_i}{x}\right|^2\right)\geq0,\quad i=1\dots m. \\
&||x||^2=1.\\
\end{aligned}
\label{eq:biPhaseweights}
\end{equation}
We relax this problem to the following maximum eigen value problem that we call weighted one bit Phase retrieval (Weighted1bitPhase).
\begin{equation}
\max_{x,||x||=1}x^* \frac{1}{m}\sum_{i=1}^m y_i\left(R^1_{i} a^1_ia^{1,*}_i-R^2_{i} a^2_ia^{2,*}_i\right)x
\label{eq:WeightOneBit}
\end{equation}

\noindent Thanks to the boundedness of $(R^1_i,R^2_i)$, one can carry the same analysis done in \ref{theory}, and get correctness and sample complexity for this formulation. 
Indeed $O(2n\log(n))$ measurements are also  sufficient for phase retrieval from that weighted scheme.
Nevertheless this scheme is more sensitive to noise and distortion than the original formulation.

\begin{theorem}[Weighted One bit Recovery]
For $x_0 \in \mathbb{C}^n, ||x_0||=1$, and $\epsilon>0$.Let $\hat{x}_m$ be the solution of problem Weighted1Bit \eqref{eq:WeightOneBit}. Then for any $\epsilon\in [0,1]$, we have with we have with a probability at least $1-O(n^{-2}) $, 
$$\text{for } m \geq \frac{C}{\epsilon^2} n\log(n), \quad ||\hat{x}_m-x_0e^{i\phi}||^{2}\leq \epsilon.$$
where $C$ is universal constant, and $\phi \in [0,2\pi]$ is a global phase. 
\label{theo:mainW1bit}
\end{theorem}
\noindent The proof of this Theorem is given in the  Section \ref{theory}. 
\begin{remark}
For simplicity of the exposure we limit the analysis to $\theta(z)=z$.
\end{remark}

\section{Greedy Refinements via Alternating Minimization }\label{sec:Greed}
We have now defined $2$ variants of one bit phase retrieval : 1BitPhase and Weighted1BitPhase.
Both formulation allows phase recovery via a spectral maximum Eigen value problem.
In this section we start by analyzing the alternating minimization approach and the virtues of the initialization step proposed in \cite{AM} that we call SubExpPhase.
We then show that $1$BitPhase and Weighted$1$BitPhase offer   a new way to initialize the alternating minimization problem.
Note that we have now $3$ randomized strategies (SubExpPhase, $1$BitPhase and Weighted$1$BitPhase ) to initialize the alternating minimization problem of phase retrieval.
Each one succeeds with high  probability, a multiple initialization strategy  allows to choose the corresponding solution with lowest MSE.
 
\subsection{ Phase Recovery via Alternating Minimization}
The alternating minimization algorithm proposed in \cite{AM} has 2 main ingredients:
\begin{enumerate}
\item For an accuracy $\epsilon$, given $O( \frac{1}{\epsilon^2}n\log^3(n))$ measurements, the authors propose an initialization $x^0$ that is an $\epsilon$ estimate of $x_0$.
\item A resampling procedure that ensures a reduction in the error in each step of the alternating minimization provided with the above initialization.
\end{enumerate}
The resulting algorithm has a sample complexity of $O(n(\log^3(n)+\log(\frac{1}{\epsilon})\log(\log(\frac{1}{\epsilon}))))$, and a computational complexity of $O(n^2(\log^3(n)+\log^2(\frac{1}{\epsilon})\log(\log(\frac{1}{\epsilon}))))$.
Thus in order to get a better sample complexity of the resulting algorithm and hence  a better computational complexity, the challenge is to propose a better initialization.
We will show that one bit phase retrieval offer a good strategy for initializing the alternating minimization.
\subsubsection{ Sub-Exponential Initialization }
 We  comment in this section  on the initialization and the alternating minimization procedure of \cite{AM}.\\
Let $b_{i}=|\scalT{a_i}{x_0}|^2,i=1\dots m$, where $a_{i}\sim \mathcal{C}\mathcal{N}(0,I_{n})$.
The initialization proposed amounts to taking  the maximum eigen vector $\hat{x}_m$ of 
\begin{equation}
\hat{C}_m=\frac{1}{m}\sum_{i=1}^m b_i a_ia_i^*,
\label{eq:subexp}
\end{equation}
In \cite{AM} authors show correctness and concentration of this approach. We restate here their main result, and give for completeness a sketch of the proof in \ref{ap:Sujay}:
\begin{theorem}[Correctness and Concentration]\label{theo:Init}
Let $x,x_0 \in \mathbb{C}^n$. Assume that $x$ and $x_0$ are unitary, and $a \sim \mathcal{C}\mathcal{N}(0,I_{n})$.
Let $$\mathcal{E}^{x_0}(x)=\mathbb{E}|\scalT{a}{x_0}|^2|\scalT{a}{x}|^2, \quad \mathcal{E}^{x_0}(x)=\frac{1}{m}\sum_{i=1}^ m b_i |\scalT{a_i}{x}|^2, \quad \hat{x}_{m}=\argmax_{x,||x||=1}\mathcal{E}^{x_0}(x).$$
We have the following claims:
\begin{enumerate}
\item $\mathcal{E}^{x_0}(x)=x^*Cx$, where $C=\mathbb{E}(b aa^*)$, where $b=|\scalT{a}{x_0}|^2$.
\item   For all $x$, such that $||x||=1$, $\mathcal{E}^{x_0}(x)=|\scalT{x_0}{x}|^2+1$.
\item  For all $x$, such that $||x||=1$, $\frac{1}{2}||xx^*-x_0x_0^*||_F^2 = \left(\mathcal{E}^{x_0}(x_0)-\mathcal{E}^{x_0}(x)\right)=(1-|\scalT{x_0}{x}|^2)$.
\item $\frac{1}{2}||\hat{x}_m\hat{x}_m^*-x_0x_0^*||^2_{F} \leq 2\left|\left|\hat{C}_m-C\right|\right|.$
\item Let $\epsilon \in [0,1]$ then $\text{For}\quad  m  \geq c \frac{  n \log^3(n)}{\epsilon^2},\quad  ||\hat{C}_m-C||\leq 2 \epsilon \text{ with probability at least } 1-O(n^{-2}).$
\item Let $ \epsilon \in [0,1]$ then $ \text{For}\quad  m \geq c \frac{  n \log^3(n)}{ \epsilon^2}, ||\hat{x}_m-x_0e^{i\phi}||^{2}\leq \epsilon  \text{ with probability at least } 1-O(n^{-2}), $ where $c$ is a universal constant.
\end{enumerate}
\end{theorem}
\subsection{Discussion: One Bit Phase Retrieval as an Initialization to the Alternating Minimization}
\noindent Note that $b_i,i=1\dots m $ are exponential random variable thus we call that initialization Sub-exponential Initialization.\\
The concentration of $\hat{C}_m$ around $C$, depends upon  the boundedness of $b_i$  and $a_i$ by the non commutative matrix Bernstein inequality (Theorem \ref{eq:Bernstein}). We have with high probability  that $$b_i||a_i||^2\leq 4\log(m)n,$$
thus we have a sample complexity of $O(n\log^3(n))$ due to the extra contribution of $b_i$ with a $\log(n)$ term.
Recall that  the solution of one bit phase retrieval and weighted one bit phase retrieval  is the maximum eigen vector of $$\hat{C}_m=\frac{1}{m}\sum_{i=1}^m y_i(a^1_i a^{1,*}_i-a^2_ia^{2,*}_i) \quad \text{ and } \hat{C}_m=\frac{1}{m}\sum_{i=1}^m y_i(R^1_ia^1_i a^{1,*}_i-R^2_ia^2_ia^{2,*}_i),$$
respectively.
The measurements $y_i,R^1_i,R^2_i$ are bounded by $1$ and  do not affect the bound . Thus,   the sample 
complexity reduces to  only $O(n\log(n))$ pairs of measurements for phase retrieval via one-bit measurements. Thus we can initialize the alternating minimization with the solution of  One  Bit Phase and  Weighted One Bit Phase and get a better sample complexity especially in high dimensions.\\

\noindent \textbf{Geometric intuition.} We see in Figure 2 that the levels sets of the objective of SubExpPhase initialization consists of a\textbf{ paraboloid} , that is symmetric, hence it intersects the unity sphere in a symmetric zone thus phase retrieval  is possible up to a global phase. Compared to one bit initialization, the levels sets are \textbf{hyperbolic paraboloids}. Hyperbolic paraboloid are more "pointy" than paraboloid thus the surface of intersection with the sphere is smaller. The above discussion gives an intuition of the reasons behind  the better sample complexity of one-bit phase retrieval.\\
\begin{figure}[ht]
\begin{center}
\noindent\includegraphics[scale=0.28]{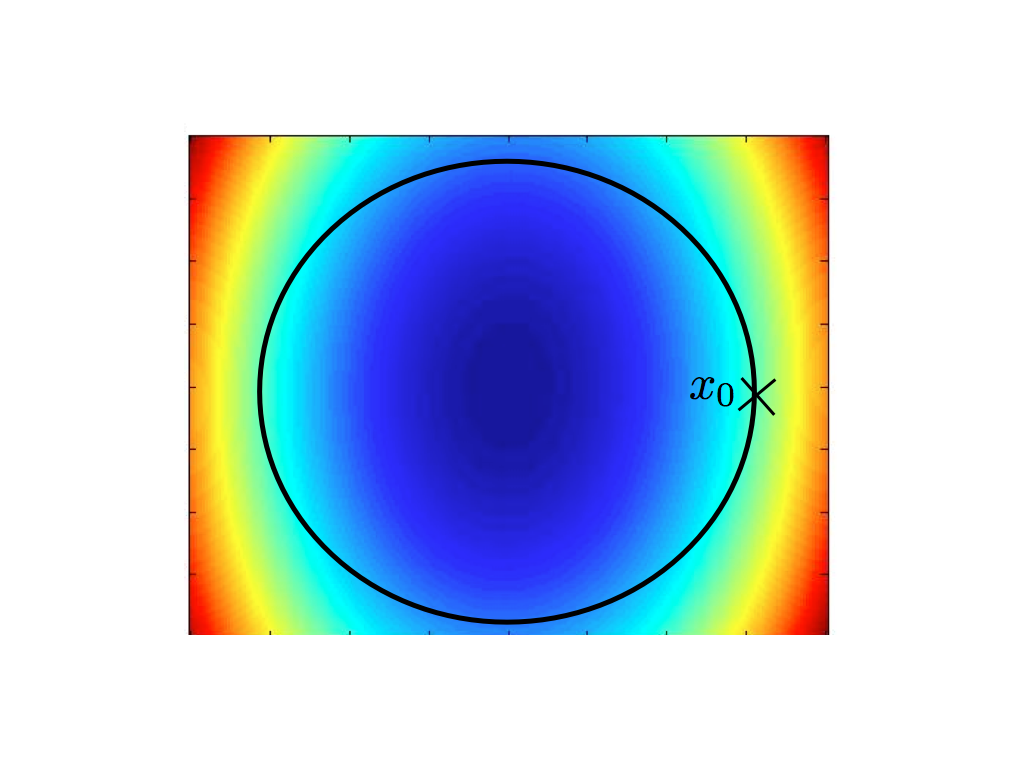}
\end{center}
\caption { The level sets of  the objective of the initialization: $x^*\frac{1}{m}\sum_{i=1}^m b_ia_ia_i^*x$, for $x_0=(1,0)$. (in red maximal values, in blue minimal values).}
\end{figure}

\noindent We conclude this section with a surprising fact:  \\
\textit{  Quantization and Greed are good. Quantization plays the role of a preconditioning that enhances the sample complexity of the initialization step of the greedy alternating minimization in phase retrieval.
}

\subsubsection{Resampling Procedure and Error Reduction }
Now given one of the three initialization strategies, namely the Sub-Exponential initialization of \cite{AM}, One Bit Phase Retrieval and Weighted One Bit Phase Retrieval.
The following algorithm proposed in \cite{AM}, proceeds in alternating the estimating of the phase and the signal. For technical reasons - mainly ensuring independence - 
the algorithm proceeds in a stage-wise alternating  minimization. At each stage we use a new re-sampled sensing matrix and the corresponding measurements.
\begin{algorithm}[H]
 \begin{algorithmic}[1]
 \Procedure{AltMinPhaseResampling}{$A,b,\epsilon$}
 \State $t_0 \gets c\log(\frac{1}{\epsilon})n$
 \State Partition $b$ and the corresponding  rows of $A$ into $t_0+1$ disjoint sets:
 $(b_0,A_0),\dots (b_{t_0},A_{t_0})$.
 \State \textbf {Initialize} $x$ using \textsc{SubExponentialPhase} or \textsc{1bitPhase}, or \textsc{Weighted1bit Phase.} 
 \For{$t=0 \dots t_0-1$} 
  \State $u_{t+1} \gets Ph(A_{t+1}x_{t})$
 \State  $ x_{t} \gets \arg\min ||A_{t+1}x-B_{t+1}u_{t+1} ||^2_{2}$
 \EndFor
 \State \textbf{return} $x_{t_0}$ 
 \EndProcedure
 \end{algorithmic}
 \caption{AltMinPhase with Resampling }
 \label{AltMinPhaseRes}
\end{algorithm}
\noindent Combining results from \cite{AM} with Theorem~\ref{theo:main} and Theorem~\ref{theo:mainW1bit} we have:\\
\begin{theorem}
For every $\epsilon>0$
Algorithm \ref{AltMinPhaseRes} outputs $x_{t_0}$ such that $||x_{t_0}-x_0e^{i\phi}||_{2}\leq \epsilon$ with high probability. 
The sample complexity depends upon the initialization step. 
\begin{enumerate}
\item Sub-Exponential Initialization: the sample complexity is  $O\left(n\log^3(n)+\log(\frac{1}{\epsilon})\log(\log(\frac{1}{\epsilon})) \right)$.
\item Weighted/One Bit Phase Retrieval: the sample complexity is  $O\left(2n(\log(n)+\log(\frac{1}{\epsilon})\log(\log(\frac{1}{\epsilon})) \right)$.
\end{enumerate}
\end{theorem}
This theorem is a consequence of the work of \cite{AM} that does not depend on the initialization step. 
The greedy refinements of one bit solution, ensures convergence to the optimum with high probability and lower sample complexity than the one obtained in \cite{AM}. 
\begin{remark}[Multiple Initialization]
Fix the total  number of measurements.
Let $x_{s}$ the solution of  Algorithm  \ref{AltMinPhaseRes} initialized with the Sub-exponential  initialization.
$x_{1b}$  the solution of  Algorithm  \ref{AltMinPhaseRes} initialized with the One Bit Phase initialization.
$x_{w1b}$ the solution of  Algorithm  \ref{AltMinPhaseRes} initialized with the Weighted  One Bit Phase initialization.
Define
$$x_{*}=\argmin_{x\in \{x_{s},x_{1b},x_{w1b}\}} MSE (x, u=Ph(Ax))= ||A x - B u||^2_{2},$$
The multiple initialization  strategy produces  $x_*$ that has the lower MSE for a given accuracy.
\end{remark}


\section{Theoretical Analysis}{\label{theory}
In this section we give the main steps of the proof of Theorems \ref{theo:main} and \ref{theo:mainW1bit} for one bit phase Retrieval and Weigthed One Bit phase Retrieval respectively.
\subsection{One Bit Phase Retrieval: Correctness and Concentration }\label{theory1bit}
In this section we state Propositions \ref{lem:compeq},\ref{lem:empiprocess}, and \ref{lem:radav} which form the core of our analysis for one bit Phase Retrieval. 
The proofs are given in Appendix \ref{app:A}.
We need the following preliminary definition.
\begin{definition}[Risk and Empirical risk]
Let $x_0 \in \mathbb{C}^n, ||x_0||=1$.
For $x \in \mathbb{C}^{n}$ such that $||x||=1$, and $A=\{a^1,a^2\}$ i.i.d. complex Gaussians, let 
\begin{equation*}
\mathcal{E}^{x_0}(x)=x^*C x, \end{equation*}
 where $C= \mathbb{E}\left(  y(a^1a^{1.*}-a^2a^{2,*})\right)$ and  $y=sign\left(\theta(|\scalT{a^1}{x_0}|^2)-\theta(|\scalT{a^2}{x_0}|^2\right)$.
 Moreover, let
\begin{equation*}
\hat{\mathcal{E}}^{x_0}(x)=x^*\hat{C}_{m} x,
\end{equation*}
where $\hat{C}_m =\frac{1}{m}\sum_{i=1}^{m}  y_i(a^{1}_ia^{1,*}_i-a^{2}_ia^{2,*}_i), y_i =Q^{\theta}_{A_i}(x_0)$ and $A_i=\{(a^1_i,a_i^2)\},i=1\dots m$ are i.i.d. complex Gaussians.
 \end{definition} 
 

We first state Proposition~\ref{lem:compeq}, that provides a theoretical justification to the relaxation introduced in the formulation 1bitPhase in \eqref{eq:MaxPhase}.
 \begin{proposition}[Correctness in Expectation]\label{lem:compeq}
 The following statements hold:
 \begin{enumerate}
\item For all $x \in \mathbb{C}^n, ||x||=1$, we have the following equality,
\begin{equation}
\mathcal{E}^{x_{0}}(x)= x^*Cx=\lambda \left| \scalT{x_0}{x}\right|^2.
\end{equation}
\item Let $y=Q^{\theta}_{A}(x_0)$, $C$ is a rank one matrix,
\begin{equation} 
C=\mathbb{E}(y(a^{1}a^{1,*}-a^{2}a^{2,*}))=\lambda x_0x_0^*.
\end{equation}
\item $x_0$ is an eigen vector of $C$ with  eigen value $\lambda$,   
\begin{equation}
Cx_0=\lambda x_0.
\end{equation}
\item The maximum eigenvector of $C$ is of the form $x_0e^{i\phi}$, where $\phi \in [0,2\pi]$.
The maximum eigen value is given by $\lambda$.
\end{enumerate}
\end{proposition}

Proposition \ref{lem:compeq} suggests that $x_0$ can be recovered up to global phase shift as the maximum eigen vector of the matrix $C$.
The Quality of the recovery of one bit phase recovery, and its sample complexity is therefore driven by how well the empirical Hermitian matrix $\hat{C}_{m}$, concentrates around its mean $C$.
A key quantity in the analysis is $\lambda$, which can be seen as a form of  signal to noise ratio. Recall that: 
\begin{equation}
\lambda=\mathbb{E}(sign\left(\theta(E_1)-\theta(E_2)\right)(E_1-E_2)), \quad E_1, E_2 \sim Exp(1) (\text{ iid }).
\end{equation} 
The following Lemma shows how the value of $\lambda$ depends on the observation model $\theta$, and how $\lambda$ relates to noise and distortion levels.\\

\begin{lemma}
The values of $\lambda$ for different observation models $\theta$ are given in the following:
\begin{enumerate}
\item Noiseless setup: $\theta(z)=z $,$ \quad \lambda=1$. 
\item Exponential Noise: $\theta(z)=z+\nu, \nu$ is an exponential random variable with variance $\sigma,\quad \lambda=\frac{1+2\sqrt{\sigma}}{(1+\sqrt{\sigma})^2} $. 
\item Poisson Noise: $\theta(z)=\mathcal{P}_{\eta}(z),\quad \mathcal{P}_{\eta}(z)= p \quad p \sim Poisson \left(\frac{z}{\eta}\right) \quad \lambda=\mathbb{E}\left(sign(S(E_1,E_2))(E_1-E_2)\right)$\\ $S(E_1,E_2))\sim Skellam(\frac{E_1}{\eta},\frac{E_2}{\eta}), \quad E_1,E_2 \sim Exp(1). $ $\lambda$ is a decreasing function in $\eta$.
\item Distortion setup: $\theta(z)=\tanh(\alpha z)$, $ \lambda= \mathbb{E}(|E_1-E_2|sign(1-\tanh(\alpha E_1)\tanh(\alpha E_2))),  E_1,E_2 \sim Exp(1).$
$\lambda$ is a decreasing function in $\alpha$.
\end{enumerate}
\label{lem:SNR}
\end{lemma} 
\noindent From Lemma $\ref{lem:SNR}$, we see that $\lambda$ achieves its maximum value $1$, in the noiseless case. $\lambda $ interestingly captures the SNR as it decreases with noise and distortion levels.
\begin{lemma}\label{lem:ineq}
For any $x \in \mathbb{C}^n$, $||x||=1$, the following equality holds:
\begin{equation}
\mathcal{E}^{x_0}(x_0)-\mathcal{E}^{x_0}(x)= \frac{\lambda}{2}||xx^*-x_0x_0^*||^2_{F}.
\end{equation}
\end{lemma}
\noindent Lemma \ref{lem:ineq} provides  a {\em comparison equality}  relating  $||xx^*-x_0x_0^*||^2_{F}$ to the {\em excess risk} $ \mathcal{E}^{x_0}(x_0)-\mathcal{E}^{x_0}(x)$. 
Then using results from empirical processes we bound the excess risk with the operator norm of $\hat{C}_m-C$. The rest of the proof uses results from matrix concentration inequalities \cite{VershyninReview} in order to bound $||\hat{C}_m-C||$.   
\begin{proposition}
The following inequalities hold for 
the solution $\hat{x}_m$  of problem \eqref{eq:MaxPhase},
\begin{eqnarray*} 
&\mathcal{E}^{x_0}(x_0)-\mathcal{E}^{x_0}(\hat{x}_m)\leq 2 \sup_{x, ||x||=1}(\hat {\mathcal{E}}^{x_0}(x)-\mathcal{E}^{x_0}(x))\\
&\sup_{x, ||x||=1}(\hat {\mathcal{E}}^{x_0}(x)-\mathcal{E}^{x_0}(x))=\left|\left| \hat{C}_m-C\right|\right|.
\end{eqnarray*}
\label{lem:empiprocess}
\end{proposition}
\noindent Finally  we  bound $\left|\left| \hat{C}_m-C\right|\right|$ using Matrix  Bernstein inequality \cite{VershyninReview}:
\begin{proposition}For $\epsilon \in [0,1]$, there exists a constant $c$ such that:
$$\text{For}\quad  m \geq \frac{c n \log(n)}{\lambda \epsilon^2}\quad  ||\hat{C}_m-C||\leq \epsilon \lambda \text{ with probability at least } 1-O(n^{-2}).$$

\label{lem:radav}
\end{proposition}

\subsection{Weighted One Bit Phase Retrieval: Correctness and Concentration}\label{weitheory}
In this section we sketch the proof architecture for results of Weigthed one bit phase retrieval.
We start first by a preliminary definition:
\begin{definition}[Risk and Empirical risk]
Let $x_0 \in \mathbb{C}^n, ||x_0||=1$.
For $x \in \mathbb{C}^{n}$ such that $||x||=1$, and $A=\{a^1,a^2\}$ i.i.d. complex Gaussians, let 
\begin{equation*}
\mathcal{E}^{x_0}(x)=x^*C x, \end{equation*}
 where $C= \mathbb{E}\left(  y(R^1 a^1a^{1.*}- R^2a^2a^{2,*})\right)$ and  $y=sign\left( b^1-b^2\right)$ , $R^1= \frac{b^1}{b^1+b^2}$,$R^2=\frac{b^2}{b^1+b^2}$, and $b^j=|\scalT{a^j}{x_0}|^2, j=1,2$..
 Moreover, let
\begin{equation*}
\hat{\mathcal{E}}^{x_0}(x)=x^*\hat{C}_{m} x,
\end{equation*}
where $\hat{C}_m =\frac{1}{m}\sum_{i=1}^{m}  y_i(R^1_ia^{1}_ia^{1,*}_i-R^2_ia^{2}_ia^{2,*}_i), y_i = sign(b^1_i-b^2_i), R^j_i=\frac{b^j_i}{b^1_i+b^2_i}, j=1,2$ and $A_i=\{(a^1_i,a_i^2)\},i=1\dots m$ are i.i.d. complex Gaussians.
 \end{definition} 
 

We first state Proposition~\ref{lem:weiCorrect}, that provides a theoretical justification to the relaxation introduced in the formulation Weighted1bitPhase in \eqref{eq:WeightOneBit}.
The proof of Proposition ~\ref{lem:weiCorrect} is given in the appendix \ref{proof:weighted}.
 \begin{proposition}[Correctness in Expectation]\label{lem:weiCorrect}
 The following statements hold:
 \begin{enumerate}
\item For all $x \in \mathbb{C}^n, ||x||=1$, we have the following equality,
\begin{equation}
\mathcal{E}^{x_{0}}(x)= x^*Cx= \frac{1}{2} \left| \scalT{x_0}{x}\right|^2+\frac{1}{2}.
\end{equation}
\item The maximum eigenvector of $C$ is of the form $x_0e^{i\phi}$, where $\phi \in [0,2\pi]$.
\end{enumerate}
\end{proposition}
Proposition \ref{lem:weiCorrect} suggests that $x_0$ can be recovered up to a global phase as a maximum eigen value of the matrix $C$.
The rest of the proof consists in proving the concentration of the empirical matrix $\hat{C}_m$ around its mean $C$. 
The proof architecture is the same presented in Section  \ref{theory1bit}.\\
The following lemma states a comparison equality that relates the excess risk to the distance to the optimum.
\begin{lemma}[Excess Risk]\label{lem:ineqweight}
For any $x \in \mathbb{C}^n$, $||x||=1$, the following equality holds:
\begin{equation}
\mathcal{E}^{x_0}(x_0)-\mathcal{E}^{x_0}(x)= \frac{1}{4}||xx^*-x_0x_0^*||^2_{F}.
\end{equation}
\end{lemma}
The Rest of the proof consists in bounding the excess risk $\mathcal{E}^{x_0}(x_0)-\mathcal{E}^{x_0}(x) $ using empirical processes tools and Non commutative  Matrix Bernstein inequality .
Note that the boundedness of $(R^1,R^2)$ simplifies at that point the analysis and the proofs are  a straightforward adaptation of the  one presented in Section \ref{theory1bit}.

\section{Computational Aspects}{\label{sec:comp}
\subsection{One bit Phase Retrieval Algorithms}
A  straightforward computation of the maximum eigenvector of the matrix $\hat{C}_m$ is expensive. 
One needs $O(n^2m)$ operations  to compute the matrix $\hat{C}_m$, that is $O(n^3\log(n))$, and then the computation of the first eigenvector requires $O(n^2)$ operations.
The total computational cost is therefore $O(n^3\log(n))+n^2)$, and is dominated by the cost of computing $\hat{C}_m$.\\

An elegant method  to avoid that overhead  is the power method. The power method allows for the computation of the maximum eigenvector without having to compute the matrix $\hat{C}_m$, this reduce drastically the computational cost to $O(nm)$ at each iteration of the power method, that is $O(n^2\log(n))$.
In the following we discuss $2$ algorithms:
\begin{enumerate}
\item \textbf{1bitPhasePower}: One bit Phase retrieval via the Power Method given in Algorithm \ref{1bitPhasePower} .
\item \textbf{Weigthed1bitPhasePower}: Weighted  One bit Phase retrieval via the Power Method given in Algorithm \ref{W1bitPhasePower}.
\end{enumerate}

\begin{algorithm}[H]
 \begin{algorithmic}[1]
 \Procedure{1bitPhasePower}{$A,y,\epsilon$}
 \State  Initialize $r_0$ at random, $j=1$.
 \While{$|| r_j-r_{j-1}|| >\epsilon$ or $j=1$}
 \State $r_j\gets \frac{1}{m}\sum_{i=1}^m y_i\left(\scalT{a^1_i}{r_{j-1}}a^1_{i} -\scalT{a^2_i}{r_{j-1}}a^2_{i}\right)$
 \State $\hat{\lambda} \gets ||r_j||$
 \State $r_j\gets \frac{r_j}{\hat{\lambda}}$
  \State $j\gets j+1$
 \EndWhile
 \State \textbf{return} $\left(\hat{\lambda},r\right)$ \Comment{$(\hat{\lambda},r)$ is an estimate of $(\lambda,x_0)$.}
 \EndProcedure
 \end{algorithmic}
 \caption{1bitPhasePower}
 \label{1bitPhasePower}
\end{algorithm}

\begin{algorithm}[H]
 \begin{algorithmic}[1]
 \Procedure{Weigthed1bitPhasePower}{$A,y,R,\epsilon$}
 \State  Initialize $r_0$ at random, $j=1$.
 \While{$|| r_j-r_{j-1}|| >\epsilon$ or $j=1$}
 \State $r_j\gets \frac{1}{m}\sum_{i=1}^m y_i\left( R^1_{i}\scalT{a^1_i}{r_{j-1}}a^1_{i} -R^{2}_i\scalT{a^2_i}{r_{j-1}}a^2_{i}\right)$
 \State $\hat{\lambda} \gets ||r_j||$
 \State $r_j\gets \frac{r_j}{\hat{\lambda}}$
  \State $j\gets j+1$
 \EndWhile
 \State \textbf{return} $\left(\hat{\lambda},r\right)$ \Comment{$(\hat{\lambda},r)$ is an estimate of $(\lambda,x_0)$.}
 \EndProcedure
 \end{algorithmic}
 \caption{Weighted 1bitPhasePower}
 \label{W1bitPhasePower}
\end{algorithm}
\subsection{ Alternating Minimization Algorithms}
Given an initialization the algorithm amount to simply solving a Least Squares that can be solved using conjugated gradient method that needs $O(mn)$ iterations.

\section{Numerical Experiments}
\subsection{Robustness to distortion}
We consider a  signal $x_0 \in \mathbb{C}^n$ which is a  a random complex Gaussian vector with i.i.d. entries of the form $x_0[j]=X+iY$, where $X,Y \sim \mathcal{N}(0,\frac{1}{2})$.
Let $n= 128$ , $\epsilon=0.25$. We consider  $r=1/(\epsilon^2) \ceil{\log(n)}=64$ and set  $m=rn$. So that the total number of measurements is $2m$.
We assume that we measure distorted (clipped) measurements according to the model:
$$b^1_i=\tanh(\alpha |\scalT{a^1_i}{x_0}|^2), \quad b^2_i=\tanh(\alpha |\scalT{a^2_i}{x_0}|^2)\quad a^1_i,a^2_i  \sim \mathcal{C}\mathcal{N}(0,I_n),i=1\dots m. $$
$\alpha$ corresponds to the level of distortion. The distortion is more severe as $\alpha$ increases. 
\begin{figure}[t]
\begin{center}
\includegraphics[width=0.5\linewidth]{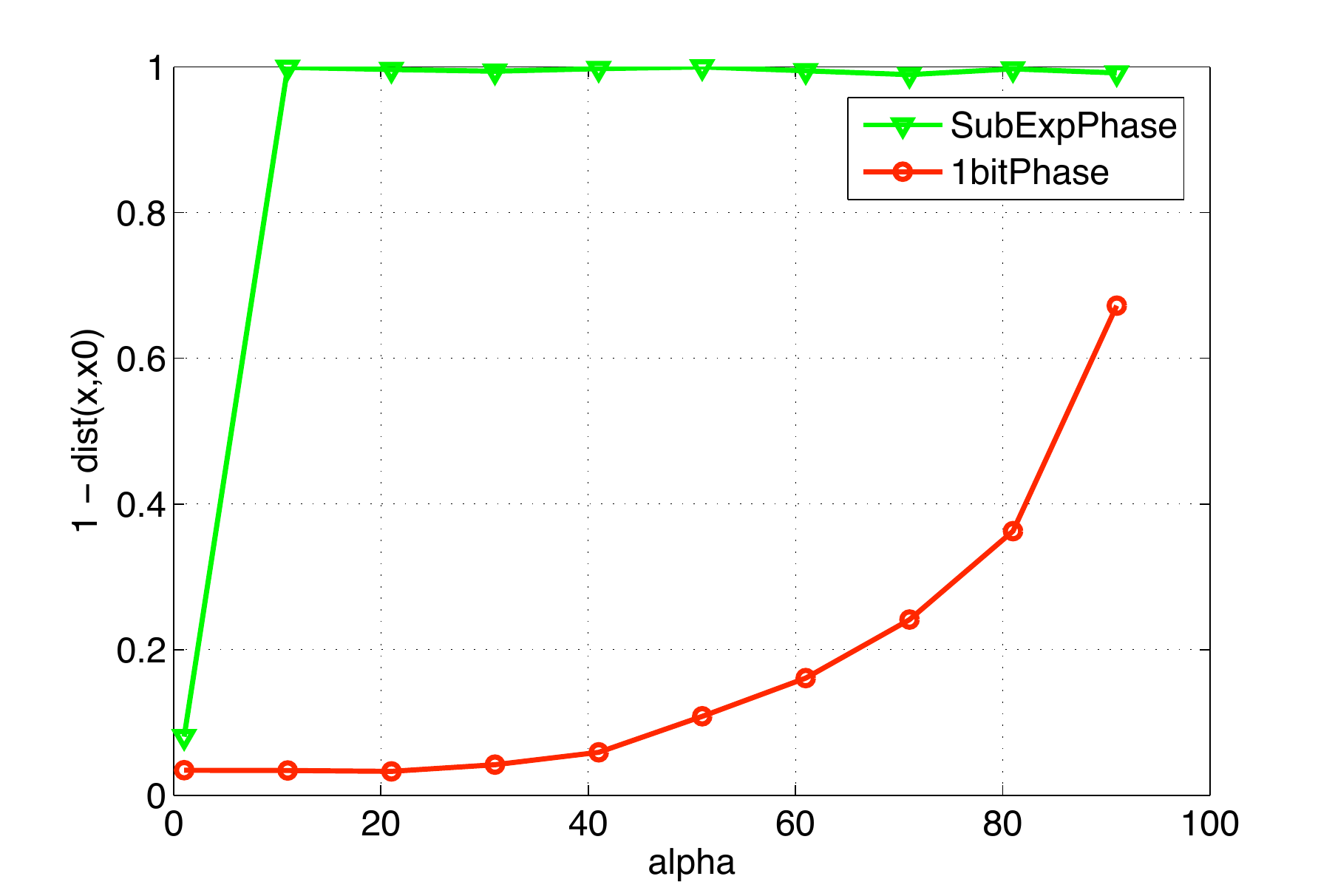}\label{fig:dist}
\end{center}
\caption{Robustness of One Bit Phase Retrieval to extreme distortion: the error of recovery $1-|\scalT{x}{x_0}|^2$ versus the distortion level $\alpha$ .}
\end{figure}
In figure \textcolor{red}{3} we plot the error of recovery  $1-dist(\hat{x}_m,x_0):=1-|\scalT{\hat{x}_m}{x_0}|^2$. Where $\hat{x}_m$ is either the solution of \textit{1bitPhase} or \textit{SubExpPhase}.
We see that one bit phase retrieval robustly recovers the signal while traditional approaches (\textit{SubExpPhase} for instance) fail under sever distortions.
\subsection{One Bit Phase Retrieval and Alternating Minimization}

We consider the problem of recovering the phase from coded diffractions patterns \cite{FMasks} or so the called Fourier masks \cite{phaseC}.
Let $F$ be the discrete Fourier Matrix,  
In this setting we measure : 
$$b=|FDiag(w)x_0|^2 \in \mathbb{R}^+_n,$$
where $w \sim \mathcal{C}\mathcal{N}(0,I_n)$ is a Gaussian complex random mask.
We generate a Gaussian random $x_0$ signal of dimension $n=8000$.
Let $m=rn$ we set $r=4$.
We measure : $$b^1_i=|\scalT{FDiag(w^1_i)}{x_0}|^2+\sigma.\max(\epsilon_i,0), b^2_i=|\scalT{FDiag(w^2_i)}{x_0}|^2+\sigma.\max(\epsilon_i,0)\quad i=1\dots r .$$
where $w^1_i,w^2_i\sim \mathcal{C}\mathcal{N}(0,I_n)$, and $\epsilon_i\sim\mathcal{N}(0,I_n)$ is a random noise.
We split the measurements in $2$ sets of size $m$,$\{b^1_i,b^2_i,i=1\dots m\}$ each and compute:
$y_i=sign(b^1_i-b^2_i)\in \{-1,1\}^n,i=1\dots r$, and $R^1_i=\frac{b^1_i}{b^1_i+b^2_i}, R^1_i=\frac{b^2_i}{b^1_i+b^2_i}$ (ratio by coordinate).
We then compute the solution of \emph{SubExpPhase},\emph{1bitPhase} and \emph{Weigthed1bitPhase}.
We then run the alternating minimization initialized with one of those solutions as well as a random initialization.
Note that all the algorithms can be now much faster thanks to the Fast Fourier transform.
In figure \ref{fig:subfig2} we see that in the noiseless setting all approaches converge, the convergence is faster for one bit variants in high dimension.
In figure \ref{fig:subfig3}\ref{fig:subfig4},\ref{fig:subfig5} we see that one bit variants are more robust in the noisy setting .
\begin{figure}[H]
\subfigure[Error $1-|\scalT{x}{x_0}|^2$ versus Iterations of AltMinPhase, for  $n=8000$, and a total measurements $8n$ in the noiseless setting.]{
\includegraphics[width=0.5\textwidth]{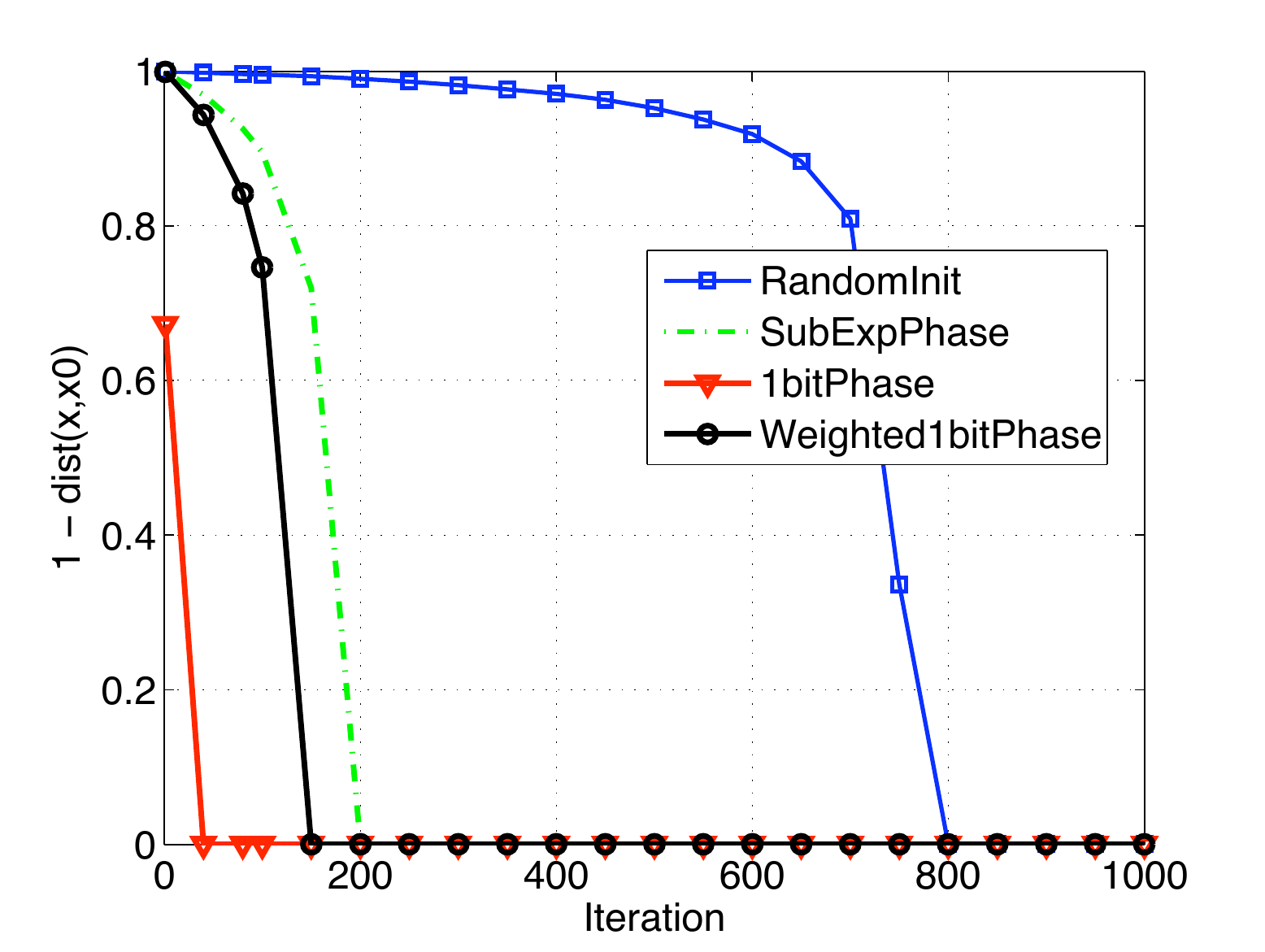}
\label{fig:subfig2}
}
\subfigure[Error $1-|\scalT{x}{x_0}|^2$ versus Iterations of AltMinPhase, for  $n=8000$ and a total measurements $8n$ in the noisy setting $\sigma=0.4$.]{
\includegraphics[width=0.5\textwidth]{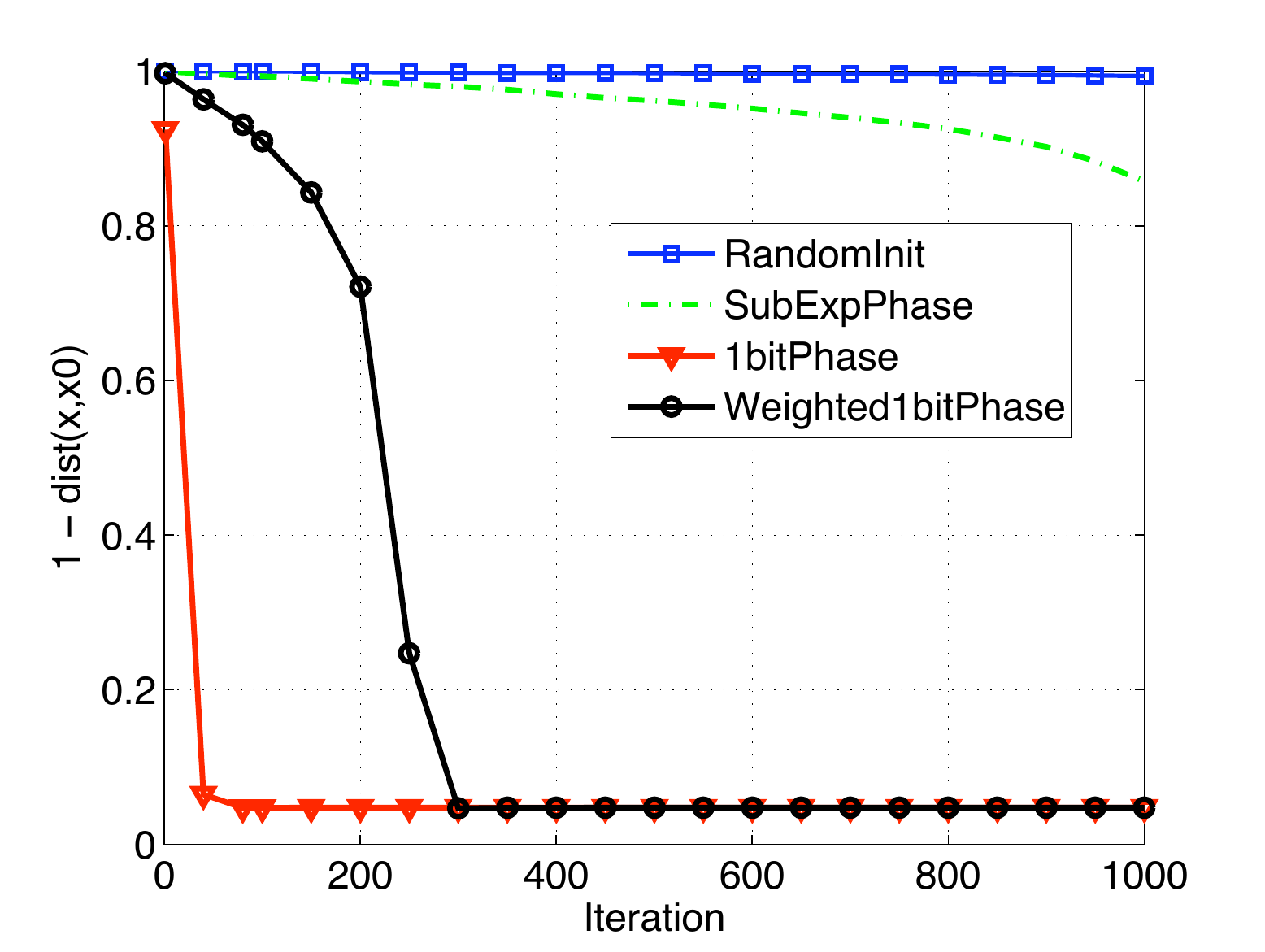}
\label{fig:subfig3}
}  
\subfigure[Error $1-|\scalT{x}{x_0}|^2$ versus Iterations of AltMinPhase, for  $n=8000$ and a total measurements $8n$  in the noisy setting $\sigma=0.8$.]{
\includegraphics[width=0.5\textwidth]{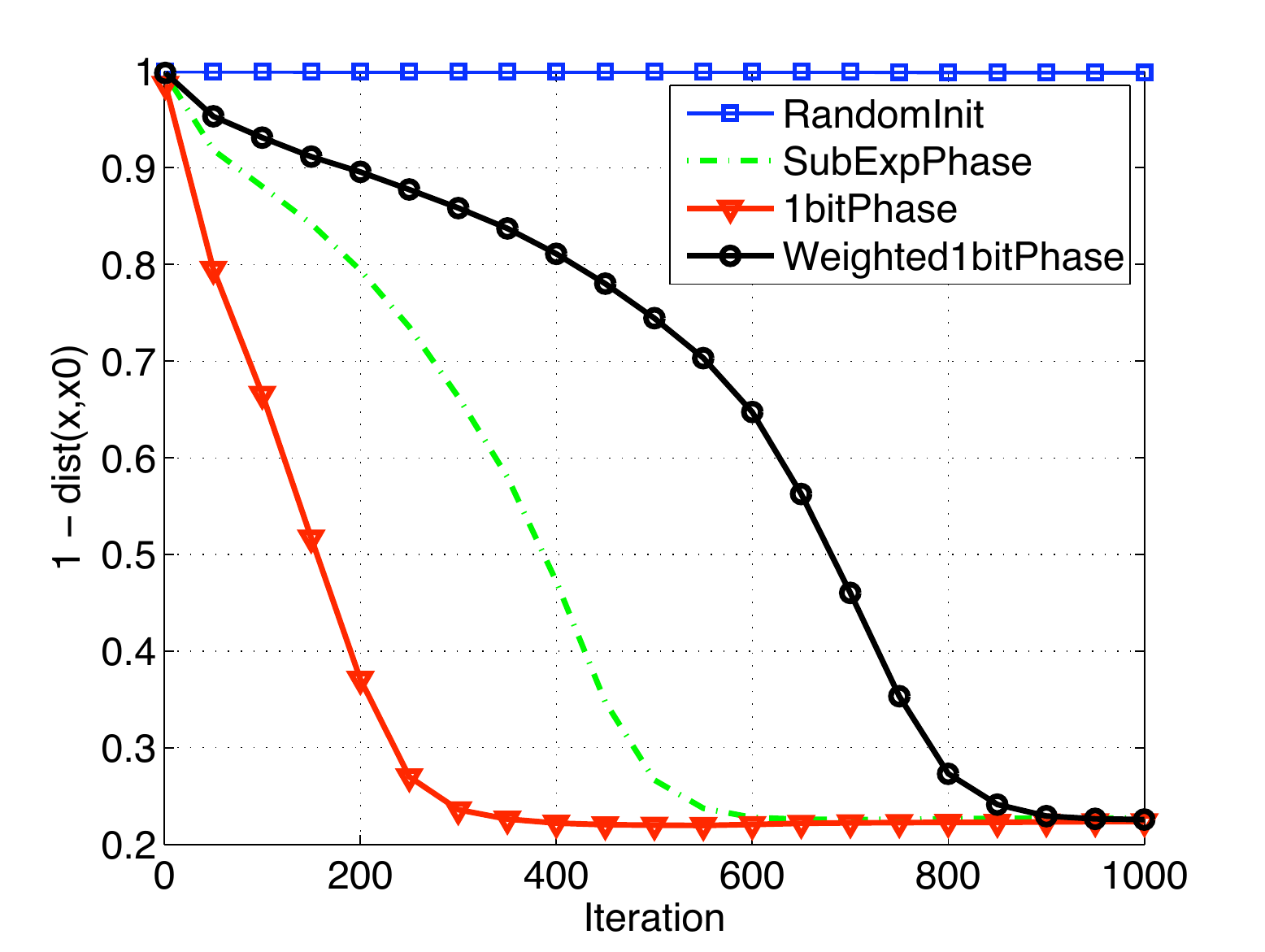}
\label{fig:subfig4}
}
\subfigure[Error $1-|\scalT{x}{x_0}|^2$ versus Iterations of AltMinPhase, for  $n=8000$ and a total measurements $8n$  in the noisy setting $\sigma=0.8$]{
\includegraphics[width=0.5\textwidth]{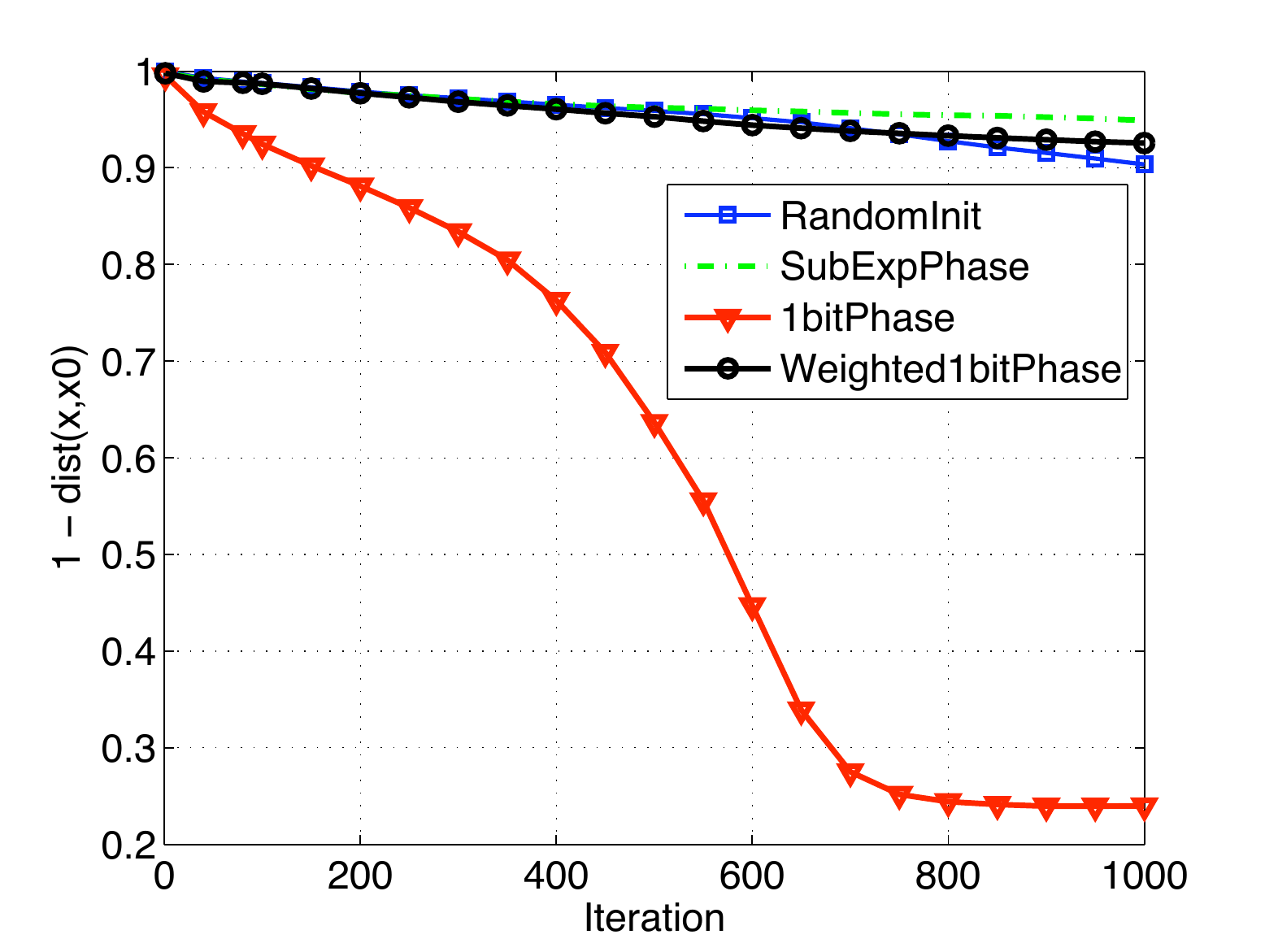}
\label{fig:subfig5}
}
\label{fig:subfigureExample}
\caption{Alternating minimization convergence with different initializations: Random Initialization,1bitPhase,Weighted1bitPhase,and SubExpPhase, in the noisy and noiseless setting.}
\end{figure}

\section{Acknowledgements}
The first author would like to thank Tomaso Poggio and Lina Mroueh for useful discussions.

\newpage

\appendix
\section{One Bit Phase Retrieval}}

\subsection{Proofs of Propositions \ref{lem:compeq},\ref{lem:empiprocess},\ref{lem:radav}} \label{app:A}

\begin{proof}[Proof of Proposition \ref{lem:compeq}]
i- For $x\in \mathbb{C}^n, ||x||=1$.
\begin{equation}
\mathcal{E}^{x_0}(x)=\mathbb{E}\left(y\left(|\scalT{a^1}{x}|^2-|\scalT{a^2}{x}|^2\right)\right),
\end{equation}
where $y=sign(|\scalT{a^1}{x_0}|^2-|\scalT{a^2}{x_0}|^2)$.
Recall $a^1,a^2\sim \mathcal{C}\mathcal{N}(0,I_n)$ are complex Gaussian vectors, there exists $g,h\sim \mathcal{N}(0,\frac{1}{2})+i \mathcal{N}(0,\frac{1}{2}) \text{ i.i.d. and }G,H  \sim \mathcal{N}(0,\frac{1}{2})+i \mathcal{N}(0,\frac{1}{2})$  i.i.d.,
\begin{eqnarray*}
\scalT{a^1}{x_0}&=& g , \quad \scalT{a^1}{x}= \scalT{x_0}{x} g +\sqrt{1-|\scalT{x_0}{x}|^2} h.\\
\scalT{a^2}{x_0}&=&  G, \quad \scalT{a^2}{x}=\scalT{x_0}{x} G  +\sqrt{1-|\scalT{x_0}{x}|^2} H.\\
|\scalT{a^1}{x}|^2-|\scalT{a^2}{x}|^2&=& \left|\scalT{x_0}{x} g +\sqrt{1-|\scalT{x_0}{x}|^2} h\right|^2- \left|\scalT{x_0}{x}G  +\sqrt{1-|\scalT{x_0}{x}|^2} H\right|^2\\
&=& |\scalT{x_0}{x}|^2(|g|^2-|G|^2)+(1-|\scalT{x_0}{x}|^2)(|h|^2-|H|^2)\\ 
&+&2 \Re\left(\overline{\scalT{x_0}{x}}\sqrt{1-|\scalT{x_0}{x}|^2}( g^*h -G^*H) \right).
\end{eqnarray*}
Recall that $y=sign(\theta(|g|^2)-\theta(|G|^2))$. From Lemma \ref{lem:exp}, we know that $E_1=|g|^2$, and $E_2=|G|^2$ are two exponential independent random variables $Exp(1)$. 
\begin{eqnarray*}
\mathcal{E}^{x_0}(x)&=&\mathbb{E}\left(y\left(|\scalT{a^1}{x}|^2-|\scalT{a^2}{x}|^2\right)\right)\\
&=& |\scalT{x_0}{x}|^2 \mathbb{E} \left(sign(\theta(E_1)-\theta(E_2))(E_1-E_2)\right)\\
&=& \lambda  |\scalT{x_0}{x}|^2.
\end{eqnarray*}
ii- Let $y= Q^{\theta}_{A} (x_0)$, by (i) we have that, $$\mathbb{E}(\scalT{y(a^1a^{1,*}-a^2a^{2,*})}{xx^*}_{F})= \lambda \scalT{xx^*}{x_0x_0^*}_{F}, \forall x ,||x||=1.$$
This means that $C=\mathbb{E}(y(a^1a^{1,*}-a^2a^{2,*}))=\lambda x_0x_0^*$. Hence $C$ is a rank one matrix.\\
iii- By (ii) we have, $Cx_0=\lambda x_0x_0^*x_0=\lambda x_0$, since $||x_0||=1$.\\
iv- By (i) $\max_{||x||=1} x^*Cx=\lambda  \max_{||x||=1} |\scalT{x_0}{x}|^2$.
It is easy to see that  $x=e^{i\phi}x_0, \phi \in [0,2\pi]$ are maximizers of the right hand side of the equation. 

\end{proof}
\begin{proof}[Proof of Lemma \ref{lem:SNR}]
i. \textit{Noiseless:} \\$\lambda= \mathbb{E}(sign(E_1-E_2)(E_1-E_2))=\mathbb{E}(|E_1-E_2|)=1$,
since $E_1-E_2 \sim Exp(1)$.\\
ii.\textit{Noisy:} \\
\textit{Exponential Noise:}\\
Let $y=sign\left((E_1+\nu_1)- (E_2+\nu_2)\right)$.
\noindent Let $L= E_1-E_2$, $L$ follows a Laplace distribution with mean 0  and scale parameter $1$: $$L \sim Laplace(0,1).$$
Let $N= \nu_1-\nu_2$, $N$ follows a Laplace distribution, $N \sim Laplace(0,\frac{1}{\gamma})$.
\noindent It follows that:
\begin{eqnarray*}
\lambda&=&\mathbb{E}_{L,N}\left(sign(L+N)L\right)\\
&=&\mathbb{E}_{L}\left(\left(1-2\mathbb{P}_{N}(N\leq -L)\right) L\right)\\
&=& \mathbb{E}_{L}\left((1-2F_{N}(-L))L\right)\\
&=& \mathbb{E}_{L}\left\{\left(1-2\left(\frac{1}{2}+\frac{1}{2}sign(-L)\left(1-\exp(-\gamma |L|)\right)\right)\right)L\right\}\\
&=& \mathbb{E}_{L} (sign(L)(1-\exp(-\gamma |L|))L)\\
&=& \mathbb{E}_{L}  |L|(1-\exp(-\gamma |L|))\\
&=& 1- \int_{0}^{+\infty}z\exp(-\gamma z) \exp(-z)dz\\
&=& 1-\frac{1}{(1+\gamma)^2} >0. 
\end{eqnarray*}

\noindent Let $\sigma=\frac{1}{\gamma^2}$ be the variance of the exponential noise.
We conclude that:
$$\lambda=\frac{1+2\sqrt{\sigma}}{(1+\sqrt{\sigma})^2}.$$
\textit{Poisson Noise}:\\
$$\lambda=\mathbb{E}(sign(p_1-p_2)(E_1-E_2)),\quad p_1|E_1\sim Poisson\left(\frac{E_1}{\eta}\right),\quad  p_2|E_2\sim Poisson\left(\frac{E_2}{\eta}\right).$$
We know that:
$$S(E_1,E_2)=p_1-p_2|\left(E_1,E_2\right) \sim Skellam\left (\frac{E_1}{\eta},\frac{E_2}{\eta}\right). $$
Hence:
$$\lambda=\mathbb{E}\left(sign(S(E_1,E_2))(E_1-E_2)\right).$$
\end{proof}
iii. \textit{Distortion:}
\begin{eqnarray*}
y &=&sign(\tanh(\alpha E_1)-\tanh(\alpha E_2))\\
&=&sign(\tanh(\alpha(E_1-E_2)))\left(1-\tanh(\alpha E_1)\tanh(\alpha E_2)\right))\\
&=&sign(\tanh(\alpha(E_1-E_2))).sign\left(1-\tanh(\alpha E_1)\tanh(\alpha E_2)\right))\\
&=& sign(E_1-E_2)sign\left(1-\tanh(\alpha E_1)\tanh(\alpha E_2)\right))
\end{eqnarray*}

$$\lambda=\mathbb{E}(y(E_1-E_2))=\mathbb{E}\left(sign\left(1-\tanh(\alpha E_1)\tanh(\alpha E_2)\right)|E_1-E_2|\right)>0.$$

\begin{proof}[Proof of Lemma \ref{lem:ineq}]
$\mathcal{E}^{x_0}(x_0)-\mathcal{E}^{x_0}(x)=\lambda(1-|\scalT{x_0}{x}|^2)=\frac{\lambda}{2}|| xx^*-x_0x_0^*||^2_{F}$, since $x_0$ and $x$ are unitary. 
\end{proof}

\begin{proof}[Proof of Proposition \ref{lem:empiprocess}]
Following the  classical approach to study empirical risk minimization in statistical learning theory we have,
$$ \mathcal{E}^{x_0}(x_0)-\mathcal{E}^{x_0}(\hat{x}_m)\leq 2 \sup_{x ,||x||=1}\left|\hat {\mathcal{E}}^{x_0}(x)-\mathcal{E}^{x_0}(x)\right|.$$
$$\hat {\mathcal{E}}^{x_0}(x)-\mathcal{E}^{x_0}(x)=\left(x^* (\hat{C}_m-C)x\right)$$
Hence:
\begin{equation}
\frac{\lambda}{2}|| \hat{x}_m \hat{x}_m^*-x_0x_0^*||^2_{F}\leq   2 \sup_{x ,||x||=1}\left| x^*(\hat{C}_m-C)x\right|=2 \left|\left| \hat{C}_m-C\right|\right| ,
\label{eq:concentration}
\end{equation}
where $\hat{C}_m=\frac{1}{m}\sum_{i=1}^m C_i$, $ C_i= y_i(\Delta_i)$,$\Delta_i=a^1_ia^{1,*}_i-a^2_ia^{2,*}_i$ and $C=\mathbb{E}(y(a^1a^{1,*}-a^{2}a^{2,*}))=\lambda x_0x_0^*$.\\
\end{proof}

\begin{proof}[Proof of Proposition \ref{lem:radav}]
Let $$X_i =\frac{1}{m}(y_i(a^1_ia^{1,*}_i-a^2_ia^{2,*}_i) -\lambda x_0x_0^*).$$

\noindent We would like to get a bound on $\left|\left|\sum_{i=1}^m X_i\right|\right|$, the main technical issue goes to the fact that $||X_i||$ are not bounded almost surely.
We will address that issue by rejecting samples outside the ball of radius $\sqrt{M}$, where $M$ is defined in the following.\\

\noindent Let $M= 2n(1+\beta)^2$.
 Let $ E=\{(a^1,a^2),  ||a^1||^2\leq M \text{ and } ||a^2||^2\leq M\}$.
Let $$(\tilde{a^1}_i,\tilde{a^2}_i)= (a^1_i,a^2_i) \text{ if } (a^1_i,a^2_i)\in E \text{ and } 0 \text { otherwise}.$$
Let $$\tilde{y}_i=sign\left(|\scalT{a^1_i}{x_0}|^2-|\scalT{a^2_i}{x_0}|^2\right)  \text{ if } (a^1_i,a^2_i)\in E \text{ and } 0 \text { otherwise}. $$
Let $$\tilde{C}_m = \frac{1}{m}\sum_{i=1}^m \tilde{y}_i (\tilde{a}^1_i \tilde{a}^{1,*}_i-\tilde{a}^2_i \tilde{a}^{2,*}_i) \quad \tilde{C}=\mathbb{E}(\tilde{C}_m).$$
Note that $\tilde{C}_m$ is the sum of bounded random variable , so that we can use non commutative  matrix Bernstein inequality given in Theorem \ref{eq:Bernstein}, in order to bound  $ \left|\left|\tilde{C_m}-\tilde{C}\right|\right|$.
On the other hand by  the triangular inequality we have:
\begin{equation}
\left|\left|C_m-C\right|\right|\leq \left|\left|C_m-\tilde{C}_m\right|\right|+ \left|\left|\tilde{C_m}-\tilde{C}\right|\right|+\left|\left|\tilde{C}-C\right|\right|
\end{equation}
\textbf{Bounding $ \left|\left|C_m-\tilde{C}_m\right|\right|$:}\\

\noindent Note that $||a||^2 \sim \chi^2_{2n}$, $||a||$ is a Lipchitz function of Gaussian with constant one. 
A Gaussian concentration bound implies;
\begin{equation}
\mathbb{P}(||a_i||^2 \geq (\sqrt{2n}+t)^2)\leq e^{-\frac{t^2}{2}}.
\end{equation} 
Setting $t=\beta\sqrt{2n}$, it follows that:
\begin{equation}
\mathbb{P}(||a_i||^2 \geq 2n(1+\beta)^2)\leq e^{-\beta^2 n}.
\end{equation} 
\begin{eqnarray*}
\mathbb{P}(\max_{i=1,\dots m, j=1,2}||a^j_i||^2>M)&\leq& 2m \mathbb{P}(||a||^2> M)\\
&\leq& 2m e^{-\beta^2 n}.\\
\end{eqnarray*}
It follows that :$$\left|\left|C_m-\tilde{C}_m\right|\right|=0 \text{  with probability at least  } 1-2me^{-\beta^2 n}.$$
\noindent \textbf{Bounding $\left|\left|\tilde{C_m}-\tilde{C}\right|\right|:$}\\

\noindent Let $$\tilde{X}_i =\frac{1}{m}\left(\tilde{y}_i(\tilde{a}^1_i\tilde{a}^{1,*}_i-\tilde{a}^2_i\tilde{a}^{2,*}_i) -\tilde{C}\right).$$
It is easy to see that $||\tilde{C}||\leq ||C||=\lambda$.
$$||\tilde{X}_i|| \leq \frac{1}{m}( ||\tilde{y}_i (\tilde{a}^1_i\tilde{a}^{1,*}_i) ||+  ||\tilde{y}_i (\tilde{a}^2_i\tilde{a}^{2,*}_i) ||+||\tilde{C}||) = \frac{1}{m}\left( ||\tilde{a}^{1}_i||^2+||\tilde{a}^{2}_i||^2+\lambda \right)\leq \frac{2M+\lambda}{m}\leq \frac{4M+\lambda}{m}=K.$$
\begin{eqnarray*}
\tilde{X}^2_i&=& \frac{1}{m^2}\left(\tilde{y}_i(\tilde{a}^1_i\tilde{a}^{1,*}_i-\tilde{a}^2_i\tilde{a}^{2,*}_i) -\tilde{C}\right)\left(\tilde{y}_i(\tilde{a}^1_i\tilde{a}^{1,*}_i-\tilde{a}^2_i\tilde{a}^{2,*}_i) -\tilde{C}\right)\\
&=& \frac{1}{m^2}\left( ||\tilde{a}^1_i||^2\tilde{a}^1_i\tilde{a}^{1,*}_i +  ||\tilde{a}^2_i||^2\tilde{a}^2_i\tilde{a}^{2,*}_i-\scalT{\tilde{a}^1_i}{\tilde{a}^2_i}\tilde{a}^1_i\tilde{a}^{2,*}_i - \scalT{\tilde{a}^2_i}{\tilde{a}^1_i}\tilde{a}^2_i\tilde{a}^{1,*}_i+\tilde{C}^2- C\tilde{C_i} -\tilde{C_i}C \right)
\end{eqnarray*}
Note that $\mathbb{E}(\scalT{a^1_i}{a^2_i}a^1_ia^{2,*}_i )=\mathbb{E}(\scalT{a^2_i}{a^1_i}a^2_ia^{1,*}_i)= I$ by independence, $\mathbb{E}(a^1_ia^{1,*}_i)=I$, and 
$\mathbb{E}(\tilde{C}\tilde{C}_i)= \tilde{C}^2$, by definition.\\
On the other hand $\mathbb{E}(\scalT{\tilde{a}^1_i}{\tilde{a}^2_i}\tilde{a}^1_i\tilde{a}^{2,*}_i )$ is zero on the off diagonal and less than one on the diagonal. 
Hence $$||\mathbb{E}(\scalT{\tilde{a}^1_i}{\tilde{a}^2_i}\tilde{a}^1_i\tilde{a}^{2,*}_i )|| \leq 1.$$
Also $\mathbb{E}(\tilde{a^1}_i\tilde{a^{1,*}}_i)$ is zero on the off diagonal and less than one on the diagonal, hence:
$$||\mathbb{E}(\tilde{a}^1_i\tilde{a}^{1,*}_i)||\leq 1.$$
It follows that $$||\mathbb{E}( ||\tilde{a}^1_i||^2\tilde{a}^1_i\tilde{a}^{1,*}_i ||\leq M ||\mathbb{E}(\tilde{a}^1_i\tilde{a}^{1,*}_i)||\leq M.$$
Taking the operator norm we have:
$$||\mathbb{E}(\tilde{X}^2_i)||\leq\frac{1}{m^2}(2M+\lambda^2+2)\leq\frac{1}{m^2}(4M+\lambda^2)  .$$
Finally: 
$$\left|\left|\sum_{i=1}^m \mathbb{E}(\tilde{X}^2_i)\right|\right|\leq m \max_{i} \mathbb{E}||\tilde{X}^2_i||\leq \frac{4M+\lambda^2}{m}=\sigma^2.$$
The rest of the proof of this part is an adaptation of the proof of Theorem \ref{covariance}  in \cite{VershyninReview} on covariance estimation of heavy tailed matrices.\\
We are now ready to apply the non commutative Bernstein's inequality:
$$\mathbb{P}\left( \left|\left|\sum_{i=1}^m \tilde{X}_i\right|\right|\geq \epsilon \right)\leq 2n.\exp\left(-c \min\left(\frac{\epsilon^2}{\sigma^2},\frac{\epsilon}{K}\right)\right) 
\leq 2n.\exp\left(-c \min\left(\frac{\epsilon^2}{4M+\lambda^2},\frac{\epsilon}{4M+\lambda}\right).m\right) $$
Clearly $\lambda \leq 1$, by definition. Hence $\lambda < 4M$.
$$\min\left(\frac{\epsilon^2}{4M+\lambda^2},\frac{\epsilon}{4M+\lambda}\right)=\frac{1}{4M}\min\left(\frac{\epsilon^2}{\lambda(\frac{\lambda}{4M}+\frac{1}{\lambda})} ,\frac{\epsilon}{1+\frac{\lambda}{4M}}\right) \geq \frac{1}{4M} \min\left(\frac{\epsilon^2}{\lambda},\epsilon\right) $$
Let $\epsilon= \max(\sqrt{\lambda}\delta,\delta^2)$, $\delta= s \sqrt{\frac{4M}{m}}$.
It follows that:\\
$$\mathbb{P}\left( \left|\left|\sum_{i=1}^m \tilde{X}_i\right|\right|\geq \epsilon \right)\leq 2n \exp\left(-c \delta^2 \frac{m}{4M}\right)=2n\exp(-cs^2).$$
Therefore we have with a probability at least $1-2n\exp(-cs^2)$ : 
$$||\tilde{C}_m-\tilde{C}||\leq \max(\sqrt{\lambda}\delta,\delta^2) \quad \delta =s\sqrt{\frac{4M}{m}},$$
Setting $s= t \sqrt{\log(n)}$, we have finally:
$$||\tilde{C}_m-\tilde{C}|| \leq \max(\sqrt{\lambda}\delta,\delta^2) \quad \delta =t\sqrt{\frac{4M\log(n)}{m}} \text{with probability at least } 1-n^{-t^2}.$$

\noindent \textbf{Bounding $\left|\left|\tilde{C}-C\right|\right|$:}\\
By the rotation invariance of Gaussian we can assume $x_0=(1,0,\dots,0)$.\\
The off diagonal terms of $\mathbb{E}(\tilde{y}(\tilde{a}^1\tilde{a}^{1,*}-\tilde{a}^2\tilde{a}^{2,*})$ are zero.
The same holds for  $\mathbb{E}(y(a^1a^{1,*}-a^2 a^{2,*})$.
The only term that is non zero  on the diagonal is first one. 
\begin{eqnarray*}
||\tilde{C}-C|| &=&\mathbb{E}(y(|a^1_1|^2-|a^2_1|^2)1_{(a^1,a^2)\notin E})\\
&\leq&\left( \mathbb{E}(y^2(|a^1_1|^2-|a^2_1|^2)^2)\right)^{\frac{1}{2}}( \mathbb{E}(1_{E^c}))^{\frac{1}{2}}\\
&=& (\mathbb{E}(|a^1_1|^4+|a^2_1|^4-2 |a^1_1|^2|a^2_2|^2) )^{\frac{1}{2}} \sqrt{\mathbb{P}(E^c)}\\
&\leq&\sqrt{2+2-2}\sqrt{2 e^{-\beta^2n}}\\
&=&2  e^{-\beta^2n/2}.
\end{eqnarray*}

\noindent \textbf{Putting all together:}\\
Setting $\beta=t=\sqrt{2}$. For $2m=\hat{c} n $, $1< \hat{c}<n$.
We have with probability at least $1-O(\frac{1}{n^2})$ , since $(1-2me^{-2n})\sim (1-O(1/n^2))$, for sufficiently large $n$:
$$||\hat{C}_m-C||\leq c \sqrt{\frac{\lambda M \log(n)}{m}}+ 2e^{-n}$$
where $M =2n(1+\sqrt{2})^2$.
There exists a constant $c',\epsilon \in [0,1]$ such that:\\
$$\text{For}\quad  m \geq \frac{c'  n \log(n)}{\lambda \epsilon^2}\quad  ||\hat{C}_m-C||\leq {\epsilon} \lambda \text{ with probability at least } 1-O(n^{-2}).$$
By equation \eqref{eq:concentration} we conclude :
$$\text{For}\quad  m \geq \frac{c^{''} n \log(n)}{\lambda \epsilon^2}\quad ||\hat{x}_m-x_0e^{i\phi}||^2\leq|| \hat{x}_m \hat{x}_m^*-x_0x_0^*||^2_{F}\leq \epsilon  \text{ with probability at least } 1-O(n^{-2}).$$
where $\phi$ is a global phase.

\end{proof}

\subsection{Technical Tools}
Here, we collect a few technical results needed in the proofs.
 \begin{lemma}[Spacing of Exponentials]
For  $x_0 \in \mathbb{C}^n$, and $a \sim \mathcal{N}(0,\frac{1}{2}I_{n})+ i\mathcal{N}(0,\frac{1}{2}I_{n})$, $E= |\scalT{a}{x_0}|^2$ follows an exponential distribution with parameter one (see \cite{ThesisLina} for a proof).\\
Moreover  \cite{Orderstatistics},   if we let  and $a^1,a^2 \text{ i.i.d. } \sim \mathcal{N}(0,\frac{1}{2}I_{n})+ i\mathcal{N}(0,\frac{1}{2}I_{n})$, let $E_1(x_0)=  |\scalT{a^1}{x_0}|^2$, and $E_2
(x_0)=  |\scalT{a^2}{x_0}|^2$, and $E^{(1)}(x_0)$ and $E^{(2)}(x_0)$, then the corresponding order statistics i.e $E^{(2)}(x_0)\geq E^{(1)}(x_0)$.
$\Delta(x_0,x_0)=  |\scalT{a^{(2)}}{x_0}|^2-   |\scalT{a^{(1)}}{x_0}|^2$ is also exponentially distributed with parameter one. $\Delta$ is called spacing of order statistics of exponentials.
\label{lem:exp}
 \end{lemma}
 \begin{lemma} [Laplace Exponential]
$U\sim \exp (\gamma)$, $V \sim \exp (\gamma)$ , $U$ and $V$ are independent then $U-V \sim Laplace(0,\frac{1}{\gamma})$. 
The CDF of $U-V$ is :
$$F_{U-V}(z)=\frac{1}{2}+\frac{1}{2}sign(z)(1-\exp(-\gamma |z|)).$$
\end{lemma}
 \begin{lemma}[Uniform Ratio \cite{ThesisLina} ]
If $X$ and $Y$ are two independent chi-square variables with $2a$ and $2b$ degrees of freedom respectively, then $Z = \frac{X}{X+Y}$
has the beta distribution with parameter $a$ and $b$, $\beta(a,b)$.\\
For $a=b=1$: 
If $X$ and $Y$ are two independent Exponential random variable with mean one  then $Z = \frac{X}{X+Y}$
is uniformly distributed $Unif[0,1]$.
\label{lem:unif}
\end{lemma}
 

\begin{theorem}[Non commutative Bernstein Inequality \cite{VershyninReview} ]\label{CovConc}
Consider a finite sequence $X_i$ of independent centered self adjoint random $n\times n$ matrices. Assume we have for some numbers $K$ and $\sigma$ that:
$$||X_i||\leq K \text{ almost surely} \quad ||\sum_{i}\mathbb{E}X^2_i||\leq \sigma^2.$$
Then, for every $t>0$, we have:
$$\mathbb{P}\{ ||\sum_{i}X_i||>t\}\leq 2n \exp\left(\frac{-t^2/2}{\sigma^2+Kt/3}\right).$$
\label{eq:Bernstein}
\end{theorem}
\begin{theorem}[Covariance Estimation for Arbitrary distributions \cite{VershyninReview}]
Consider a distribution with covariance matrix $\Sigma$ supposed in some centered ball whose radius we denote $\sqrt{M}$.
Let $\Sigma_m$ be the empirical covariance. 
Let $\epsilon \in [0,1]$, and $t \geq 1$. Then the following holds with probability at least $1-n^{-t^2}$:
$$\text{ If } m\geq C (t/\epsilon)^2 ||\Sigma||^{-1}M\log(n) \text{ then } ||\Sigma_m-\Sigma||\leq \epsilon ||\Sigma||. $$
\label{covariance}
\end{theorem}

\section{Weighted one Bit Phase Retrieval}
\begin{proof}[Proof of Proposition \ref{lem:weiCorrect}]
Recall that $\mathcal{E}^{x_0}(x)=x^*C x, $where  $C=\mathbb{E}\left(y (R^1a^1a^{1,*}-R^2a^{2}a^{2,*})\right).$ 
where $y=sign(b^1-b^2)$, $R^1= \frac{b^1}{b^1+b^2}$,$R^2=\frac{b^2}{b^1+b^2}$, and $b^j=|\scalT{a^j}{x_0}|^2, j=1,2$.
Recall $a^1,a^2\sim \mathcal{C}\mathcal{N}(0,I_n)$ are complex Gaussian vectors, there exists $g,h\sim \mathcal{N}(0,\frac{1}{2})+i \mathcal{N}(0,\frac{1}{2}) \text{ i.i.d. and }G,H  \sim \mathcal{N}(0,\frac{1}{2})+i \mathcal{N}(0,\frac{1}{2})$  i.i.d.,
\begin{eqnarray*}
\scalT{a^1}{x_0}&=& g , \quad \scalT{a^1}{x}= \scalT{x_0}{x} g +\sqrt{1-|\scalT{x_0}{x}|^2} h.\\
\scalT{a^2}{x_0}&=&  G, \quad \scalT{a^2}{x}=\scalT{x_0}{x} G  +\sqrt{1-|\scalT{x_0}{x}|^2} H.\\
R^1|\scalT{a^1}{x}|^2-R^2|\scalT{a^2}{x}|^2&=& R^1\left|\scalT{x_0}{x} g +\sqrt{1-|\scalT{x_0}{x}|^2} h\right|^2-R^2 \left|\scalT{x_0}{x}G  +\sqrt{1-|\scalT{x_0}{x}|^2} H\right|^2\\
&=& |\scalT{x_0}{x}|^2(R^1|g|^2-R^2|G|^2)+(1-|\scalT{x_0}{x}|^2)(R^1|h|^2-R^2|H|^2)\\ 
&+&2 \Re\left(\overline{\scalT{x_0}{x}}\sqrt{1-|\scalT{x_0}{x}|^2}(R^1 g^*h -R^2G^*H) \right)\\
&=&  |\scalT{x_0}{x}|^2(\frac{b^1}{b^1+b^2}b^1-\frac{b^2}{b^1+b^2}b^2)+(1-|\scalT{x_0}{x}|^2)(R^1|h|^2-R^2|H|^2)\\ 
&+&2 \Re\left(\overline{\scalT{x_0}{x}}\sqrt{1-|\scalT{x_0}{x}|^2}(R^1 g^*h -R^2G^*H) \right)\\
&=& |\scalT{x_0}{x}|^2(b^1-b^2)+(1-|\scalT{x_0}{x}|^2)(R^1|h|^2-R^2|H|^2)\\ 
&+&2 \Re\left(\overline{\scalT{x_0}{x}}\sqrt{1-|\scalT{x_0}{x}|^2}(R^1 g^*h -R^2G^*H) \right).
\end{eqnarray*}
Taking the expectation we get finally: 
\begin{eqnarray*}
\mathcal{E}^{x_0}(x)&=&\mathbb{E}(y(R^1|\scalT{a^1}{x}|^2-R^2|\scalT{a^2}{x}|^2))\\
&=& |\scalT{x_0}{x}|^2 \mathbb{E}(y(b_1-b_2))+(1-|\scalT{x_0}{x}|^2) \mathbb{E}(y(R^1-R^2))\\
&=&  \mathbb{E}(|b^1-b^2|-|R^1-R^2|)|\scalT{x_0}{x}|^2 +\mathbb{E}(|R^1-R^2|)\\
&=& \frac{1}{2}(|\scalT{x_0}{x}|^2+1).
\end{eqnarray*}
Where the first equality follows from independence and that $h$ and $H$ have variance one. 
The last equality holds since $|b^1-b^2|$ is  exponentially distributed with mean one.
Note also that $R^1=U$ , where $U\sim Unif[0,1]$,  and $R^2=1-U$.
Hence $|R^1-R^2|= |2U-1|$, and $\mathbb{E}(|R^1-R^2)=\mathbb{E}|2U-1|=\frac{1}{2}.$

\label{proof:weighted}
\end{proof}
\section{Sub-Exponential Initialization}
 
 \begin{proof}[Proof of Theorem \ref{theo:Init}]
\begin{eqnarray*}
\mathcal{E}^{x_0}(x)&=&\mathbb{E}(|g|^2(|\scalT{x_0}{x}|^2 |g|^2+(1-|\scalT{x_0}{x}|^2)|h|^2+2\sqrt{1-|\scalT{x_0}{x}|^2}Re({\scalT{x}{x_0})}^*g^*h)\\
&=& |\scalT{x_0}{x}|^2 \mathbb{E}(|g|^4)+ (1-|\scalT{x_0}{x}|^2)\mathbb{E}(|g|^2)\mathbb{E}(|h|^2)\\
&=& 2|\scalT{x_0}{x}|^2 +(1-|\scalT{x_0}{x}|^2)\\
&=& |\scalT{x_0}{x}|^2+1.
\end{eqnarray*}

The rest of the proof concerns the concentration of $\hat{C}_m$ around $C$ is a simple application of Theorem \ref{covariance}.
Note that $C=\mathbb{E}(baa^*)$ we have $||C||=2$. $b$ is an exponential random variable. 
$$b||a||^2 \leq 4\log(m)n, $$
with high probability, thus we can apply Theorem \ref{covariance}, and get a sample complexity bound. 
\label{ap:Sujay}
\end{proof}

%

\newcommand{\etalchar}[1]{$^{#1}$}

\end{document}